\documentclass[10pt]{article}
\usepackage[T1]{fontenc}
\usepackage{geometry} 
\usepackage{graphicx}
\usepackage{amsmath,amssymb,amsthm,amsfonts,dsfont,mathtools}
\usepackage{enumitem}
\usepackage{relsize}
\usepackage[margin=1cm]{caption}
\usepackage{wrapfig}
\usepackage{frame,color}
\usepackage{environ,subfig}
\usepackage{hyperref}
\usepackage{xspace}
\usepackage{framed}
\usepackage{comment}
\usepackage{changepage}
\usepackage[linesnumbered,boxed]{algorithm2e}
\usepackage{algorithmic}
\usepackage{mathrsfs}

\makeatletter
\renewenvironment{framed}{%
 \def\FrameCommand##1{\hskip\@totalleftmargin
 \fboxsep=\FrameSep\fbox{##1}
     \hskip-\linewidth \hskip-\@totalleftmargin \hskip\columnwidth}%
 \MakeFramed {\advance\hsize-\width
   \@totalleftmargin\z@ \linewidth\hsize
   \@setminipage}}%
 {\par\unskip\endMakeFramed}
\makeatother

\newtheorem{theorem}{Theorem}
\newtheorem{lemma}{Lemma}

\newtheorem{definition}{Definition}

\newtheorem{remark}{Remark}

\newcommand{\ignore}[1]{}

\newcommand{\Expect}{\operatorname{E}}

\newcommand{\Prob}{\operatorname{Pr}}

\newcommand{\paren}[1]{\mathopen{}\left( #1 \right)\mathclose{}}

\newcommand{\sqbrack}[1]{\left[ #1 \right]}

\newcommand{\fr}[2]{\mbox{$\frac{#1}{#2}$}}

\newcommand{\poly}{{\operatorname{poly}}}
\newcommand{\polylog}{\poly \log}

\newcommand{\outdeg}{\operatorname{outdeg}}
\newcommand{\Nout}{N_{\operatorname{out}}}
\newcommand{\Nstar}{N^\star}
\newcommand{\dist}{\operatorname{dist}}

\newcommand{\SM}{\mathsf{S}}
\newcommand{\MD}{\mathsf{M}}
\newcommand{\LG}{\mathsf{L}}

\newcommand{\sparse}{\mathsf{sp}}
\newcommand{\sparseexp}{\mathsf{s}}
\newcommand{\denseexp}{\mathsf{d}}
\newcommand{\bad}{\mathsf{bad}}

\newcommand{\rb}[2]{\raisebox{#1 mm}[0mm][0mm]{#2}}
\newcommand{\istrut}[2][0]{\rule[- #1 mm]{0mm}{#1 mm}\rule{0mm}{#2 mm}}

\newcommand{\hcm}[1][1]{\hspace*{#1 cm}}

\newcommand{\Det}{\mathsf{Det}}
\newcommand{\Detd}{\Det_{\scriptscriptstyle d}}

\newcommand{\ID}{\operatorname{ID}}
\newcommand{\LOCAL}{\mathsf{LOCAL}}

\newcommand{\dense}{{\sf DenseColoringStep}}
\newcommand{\oneshot}{{\sf OneShotColoring}}
\newcommand{\trial}{{\sf ColorBidding}}

\NewEnviron{fcodebox}[1]%
{\fbox{%
\begin{minipage}{#1\linewidth}
\begin{codebox}
\BODY
\end{codebox}
\end{minipage}
}}

\title{An Optimal Distributed $(\Delta+1)$-Coloring Algorithm?\thanks{A preliminary version of this paper appeared in {\em Proceedings 50th Annual ACM SIGACT Symposium on Theory of Computing (STOC)}, pages 445--456, 2018.}}

\author{Yi-Jun Chang\\
University of Michigan
\and
Wenzheng Li\\
IIIS, Tsinghua University
\and
Seth Pettie\thanks{This work is supported by NSF grants CCF-1514383, CCF-1637546, and CCF-1815316.}\\
University of Michigan
}

\begin{document}
\date{}
\maketitle
\thispagestyle{empty}
\setcounter{page}{0}

\begin{abstract}
{\em Vertex coloring} is one of the classic symmetry breaking problems studied in distributed computing.
In this paper we present a new algorithm for $(\Delta+1)$-list coloring in the randomized $\LOCAL$
model running in $O(\Detd(\polylog n))$ time, 
where $\Detd(n')$ is the deterministic complexity
of $(\deg+1)$-list coloring on $n'$-vertex graphs.
(In this problem, each $v$ has a palette of size $\deg(v)+1$.)
This improves upon a previous randomized algorithm of Harris, Schneider, and Su [STOC 2016, \emph{J. ACM} 2018] 
with complexity $O(\sqrt{\log \Delta} + \log\log n + \Detd(\polylog n))$, and, for some range of $\Delta$, 
is much faster than the best known deterministic algorithm of Fraigniaud, Heinrich, and Kosowski [FOCS 2016] 
and
Barenboim, Elkin, and Goldenberg [PODC 2018], 
with complexity
$O(\sqrt{\Delta\log \Delta}\log^\ast \Delta + \log^* n)$.

Our algorithm \underline{\emph{appears to be}} optimal, 
in view of the
$\Omega(\Det(\polylog n))$ randomized lower bound due to 
Chang, Kopelowitz, and Pettie [FOCS 2016, \emph{SIAM J.~Comput.} 2019],
where $\Det$ is the deterministic complexity of $(\Delta+1)$-list coloring.
At present, the best upper bounds on $\Detd(n')$ and $\Det(n')$ 
are both $2^{O\left(\sqrt{\log n'}\right)}$ and use a black box application of network decompositions 
due to Panconesi and Srinivasan [\emph{J. Algorithms} 1996]. It is quite possible that the true deterministic complexities of both problems are the same, asymptotically, which would imply the randomized optimality of our $(\Delta+1)$-list coloring algorithm.
\end{abstract}
\newpage

\section{Introduction}\label{sect:intro}

Much of what we know about the $\LOCAL$ model has emerged from studying the complexity of four
canonical symmetry breaking problems and their variants:
maximal independent set (MIS), $(\Delta+1)$-vertex coloring,
maximal matching, and $(2\Delta-1)$-edge coloring.  The palette sizes ``$\Delta+1$'' and ``$2\Delta-1$''
are minimal to still admit a greedy sequential solution;
here $\Delta$ is the maximum degree of any vertex.

Early work~\cite{Linial92,Naor91,AwerbuchGLP89,PanconesiS96,Luby86,ABI86}
showed that
all the problems are reducible to MIS,
all four problems require $\Omega(\log^* n)$ time,
even with randomization;
all can be solved in $O(\poly(\Delta) + \log^* n)$ time
(optimal when $\Delta$ is constant),
or in $2^{O(\sqrt{\log n})}$ time for any $\Delta$.
Until recently, it was actually consistent with known results that all four problems had the same complexity.

Kuhn, Moscibroda, and Wattenhofer (KMW)~\cite{KuhnMW16} proved that the ``independent set'' problems (MIS and maximal matching)
require $\Omega\paren{\min\left\{\fr{\log\Delta}{\log\log\Delta},\, \sqrt{\fr{\log n}{\log\log n}}\right\}}$ time,
with or without randomization,
via a reduction from $O(1)$-approximate minimum vertex cover.
This lower bound provably separated MIS/maximal matching
from simpler symmetry-breaking problems like $O(\Delta^2)$-coloring, which can be solved in $O(\log^* n)$ time~\cite{Linial92}.

We now know the KMW lower bounds cannot be extended to the canonical coloring problems, nor to variants of MIS like
$(2,t)$-ruling sets, for $t\ge 2$~\cite{BishtKP14,BEPS16,Ghaffari16}.
Elkin, Pettie, and Su~\cite{ElkinPS15} proved that $(2\Delta-1)$-list edge coloring
can be solved by a randomized algorithm in $O(\log\log n + \Det(\polylog n))$ time, which shows that neither
the $\Omega\paren{\fr{\log\Delta}{\log\log\Delta}}$
nor $\Omega\paren{\sqrt{\fr{\log n}{\log\log n}}}$ KMW lower bound applied to this problem.
Here $\Det(n')$ represents the \emph{deterministic} complexity of the problem in question on $n'$-vertex graphs.
Improving on~\cite{BEPS16,SchneiderW10}, Harris, Schneider, and Su~\cite{HarrisSS18} proved a similar
separation for $(\Delta+1)$-vertex coloring.
Their randomized algorithm solves the problem in $O(\sqrt{\log \Delta} + \log\log n + \Detd(\polylog n))$ time,
where $\Detd$ is the complexity of $(\deg+1)$-list coloring.

The ``$\Det(\polylog n)$''  terms in the running times of~\cite{ElkinPS15,HarrisSS18} are a
consequence of the \emph{graph shattering}
technique applied to distributed symmetry breaking.
Barenboim, Elkin, Pettie, and Schneider~\cite{BEPS16} showed that all the classic
symmetry breaking problems could be reduced in $O(\log\Delta)$ or $O(\log^2\Delta)$ time, w.h.p.,
to a situation where we have independent subproblems of size $\polylog(n)$, which can then be solved with the best available
deterministic algorithm.\footnote{In the case of MIS, the subproblems actually have size $\poly(\Delta)\log n$, but satisfy
the additional property that they contain distance-$5$ dominating sets of size $O(\log n)$,
which is often just as good as having $\polylog(n)$ size. See~\cite[\S 3]{BEPS16} or \cite[\S 4]{Ghaffari16} for more discussion of this.}
Later, Chang, Kopelowitz, and Pettie (CKP)~\cite{ChangKP19} gave a simple proof illustrating \emph{why}
graph shattering is inherent to the $\LOCAL$ model: the randomized complexity of any locally checkable problem\footnote{See~\cite{Naor91,ChangP19,ChangKP19}
for the formal definition of the class of locally checkable labeling (LCL) problems.}
is at least its deterministic complexity on $\sqrt{\log n}$-size instances.

The CKP lower bound explains why the state-of-the-art randomized symmetry breaking algorithms have such strange stated running times:
they all depend on a randomized graph shattering routine (Rand.) and a deterministic (Det.) algorithm.
\begin{itemize}
\item
$O(\log\Delta + 2^{O(\sqrt{\log\log n})})$ for MIS \hfill (Rand. due to~\cite{Ghaffari16} and Det. to~\cite{PanconesiS96}),
\item
$O(\sqrt{\log\Delta} + 2^{O(\sqrt{\log\log n})})$ for $(\Delta+1)$-vertex coloring \hfill (Rand. due to~\cite{HarrisSS18} and Det. to~\cite{PanconesiS96}),
\item
$O(\log\Delta + (\log\log n)^3)$ for maximal matching \hfill (Rand. due to~\cite{BEPS16} and Det. to~\cite{FischerG17}),
\item
$O((\log\log n)^6)$ for $(2\Delta-1)$-edge coloring \hfill (Rand. due to~\cite{ElkinPS15} and Det. to~\cite{FischerGK17,GhaffariHK18}).
\end{itemize}
In each, the term that depends on $n$ is the complexity of the best
deterministic algorithm, scaled down to $\polylog(n)$-size instances.
In general, improvements in the deterministic complexities of these problems imply improvements to their randomized complexities,
but only if the running times are improved in terms of ``$n$'' rather than ``$\Delta$.''
For example, a recent line of research has improved
the complexity of $(\Delta+1)$-coloring in terms of $\Delta$, from $O(\Delta+\log^* n)$~\cite{BarenboimEK14},
to $\tilde{O}(\Delta^{3/4}) + O(\log^* n)$~\cite{Barenboim15},
to the state-of-the-art bound
of $O(\sqrt{\Delta\log\Delta}\log^\ast \Delta + \log^* n)$
due to Fraigniaud, Heinrich, and Kosowski~\cite{FraigniaudHK16},
as improved by Barenboim, Elkin, and Goldenberg~\cite{BarenboimEG18}.
These improvements do not have consequences for randomized coloring algorithms using
graph shattering~\cite{BEPS16,HarrisSS18} since we can only assume $\Delta = (\log n)^{\Omega(1)}$ in the shattered instances.
See Table~\ref{table:vertexcoloring} for a summary of lower and upper bounds for distributed $(\Delta+1)$-list coloring
in the $\LOCAL$ model.

\begin{table}\label{c}
\footnotesize
\centering
\begin{tabular}{|l|l|l|}
\multicolumn{1}{l}{} & \multicolumn{1}{l}{\bf Randomized}                   & \multicolumn{1}{l}{\bf Deterministic}\\\hline
             		& $O(\Detd(\polylog n))$    \hfill {\bf new} 	        & $O(\sqrt{\Delta \log \Delta}\log^{\ast}\Delta + \log^\ast n)$   \hfill\istrut[1.5]{4}
             	 \cite{BarenboimEG18}   \\\cline{2-3}
            & $O(\sqrt{\log\Delta}+ \log\log n + \Detd(\polylog n))$ \hcm  \hfill \cite{HarrisSS18} & $O(\sqrt{\Delta}\log^{5/2} \Delta + \log^\ast n)$ \hfill\istrut[1.5]{4}\cite{FraigniaudHK16}\\\cline{2-3}
            & $O(\log\Delta+\Detd(\polylog n))$  \hfill \cite{BEPS16}  	  &   $O(\Delta^{3/4}\log\Delta+\log^*n)$   \hcm \hfill \istrut[1.5]{4}\cite{Barenboim15}\\\cline{2-3}
        	& $O(\log\Delta+\sqrt{\log n})$  \hfill  \cite{SchneiderW10}  	&	$O(\Delta+\log^*n)$         		 \hfill \istrut[1.5]{4}\cite{BarenboimEK14}  \\\cline{2-3}
{Upper}		& $O(\Delta\log\log n)$                 \hfill \cite{KuhnW06}  &	$O(\Delta\log\Delta+\log^*n)$  	\hfill \istrut[1.5]{4}\cite{KuhnW06}  \\\cline{2-3}
{Bounds}	& $O(\log n)$        \hfill \cite{Luby86,ABI86,Johansson99}    &	$O(\Delta\log n)$			\hfill \istrut[1]{4}\cite{AwerbuchGLP89} \\\cline{2-3}
            &                                                              &	$O(\Delta^2+\log^*n)$  		\hfill \istrut[1.5]{4}\cite{GPS87,Linial92}  \\\cline{3-3}
            &                                                               &	$O(\Delta^{O(\Delta)}+\log^*n)$	 \hfill\istrut[1.5]{4}\cite{GS87}   \\\cline{3-3}	
            &                                                               &	$2^{O(\sqrt{\log n})}$      		 \hfill \istrut[1.5]{4}\cite{PanconesiS96}  \\\cline{3-3}
			&											                	&	$2^{O(\sqrt{\log n\log\log n})}$ 	\hfill \istrut[1.5]{4}\cite{AwerbuchGLP89}\\\hline\hline
{Lower}  		& $\Omega(\log^* n) $          \hfill    \cite{Naor91}      &                \rb{-3}{$\Omega(\log^* n)$}    \hfill\istrut[1.5]{4}\rb{-3}{\cite{Linial92}}    \\\cline{2-2}
{Bounds}		& $\Omega(\Det(\sqrt{\log n}))$  \hfill \cite{ChangKP19}    &										\istrut[1.5]{4}	\\\hline
\end{tabular}
\caption{\label{table:vertexcoloring}Development of lower and upper bounds for distributed $(\Delta+1)$-list coloring
in the $\LOCAL$ model.
The terms $\Det(n')$ and $\Detd(n')$ are the deterministic complexities
of $(\Delta+1)$-list coloring and $(\deg+1)$-list coloring on $n'$-vertex graphs.
All algorithms listed, except for~\cite{HarrisSS18} and ours, also solve the $(\deg+1)$-list coloring problem.}
\end{table}

\section{Technical Overview}
In the distributed $\LOCAL$ model, the undirected input graph $G=(V,E)$
and communications network are identical.
Each $v\in V$ hosts a processor that initially knows $\deg(v)$,
a unique $\Theta(\log n)$-bit $\ID(v)$,
and global graph parameters $n = |V|$ and $\Delta =\max_{v\in V} \deg(v)$.
In the $(\Delta+1)$-list coloring problem
each vertex $v$ also has a palette $\Psi(v)$ of allowable colors,
with $|\Psi(v)|\ge \Delta+1$. As vertices progressively commit
to their final color, we also use $\Psi(v)$ to denote $v$'s
available palette, excluding colors taken by its neighbors in $N(v)$.
Each processor
is allowed unbounded computation and has access to a private stream of unbiased random bits.
\emph{Time} is partitioned into synchronized rounds of communication, in which each processor
sends one unbounded message to each neighbor.
At the end of the algorithm, each $v$ declares its output label,
which in our case is a color from $\Psi(v)$ that is distinct from colors declared by all neighbors in $N(v)$.
Refer to~\cite{Linial92,Peleg00} for more on the $\LOCAL$ model and variants.

In this paper we prove that $(\Delta + 1)$-list coloring can be
solved in $O(\Detd(\polylog n))$ time w.h.p.
Our algorithm's performance is best contrasted with the $\Omega(\Det(\polylog n))$ randomized lower
bound of~\cite{ChangKP19}, where $\Det$ is the deterministic complexity of $(\Delta+1)$-list coloring.
Despite the syntactic similarity between the $(\deg+1)$- and $(\Delta+1)$-list coloring problems,
there is no hard evidence showing their complexities are the same, asymptotically.  On the other hand,
every deterministic algorithmic technique developed for $(\Delta+1)$-list coloring applies equally well to
$(\deg+1)$-list coloring~\cite{FraigniaudHK16,BarenboimEG18,Barenboim15,PanconesiS96,AwerbuchGLP89}.
In particular, there is only \emph{one} tool that yields upper bounds in terms of
$n$ (independent of $\Delta$), and that is network decompositions~\cite{AwerbuchGLP89,PanconesiS96}.

Intellectually, our algorithm builds on a succession of breakthroughs by
Schneider and Wattenhofer~\cite{SchneiderW10},
Barenboim, Elkin, Pettie, and Schneider~\cite{BEPS16},
Elkin, Pettie, and Su,~\cite{ElkinPS15}, and
Harris, Schneider, and Su~\cite{HarrisSS18},
which we shall now review.

\subsection{Fast Coloring using Excess Colors}

Schneider and Wattenhofer~\cite{SchneiderW10} gave the first evidence that
the canonical coloring problems may not be subject to the KMW lower bounds.  They showed that for any constants $\epsilon > 0$ and $\gamma > 0$, when $\Delta \geq \log^{1 + \gamma} n$ and the palette size is $(1+\epsilon)\Delta$,
vertex coloring can be solved w.h.p.~in just $O(\log^\ast n)$ time~\cite[Corollary~14]{SchneiderW10}.
The emergence of this log-star behavior in \cite{SchneiderW10} is quite natural.
Consider the case where the palette size of each vertex is
at least $k \Delta$, where $k \ge 2$.
Suppose each vertex $v$ selects $k /2$ colors at random from its palette.
A vertex $v$ can successfully color itself if one of its selected
colors is not selected by any neighbor in $N(v)$.
The total number of colors selected by vertices in $N(v)$ is at most $k \Delta / 2$.
Therefore, the probability that a color selected by $v$ is also selected by someone in $N(v)$ is at most $1/2$,
so $v$ successfully colors itself with probability at least $1 - 2^{-k/2}$.
In expectation, the degree of any vertex after this coloring procedure is at most
$\Delta' = \Delta/2^{k/2}$.
In contrast, the number of
\emph{excess colors}, i.e.,
the current available palette size minus the number of uncolored neighbors,
is non-decreasing over time.
It is still at least
$(k-1)\Delta = (k-1)2^{k/2}\Delta'$
Intuitively, repeating the above procedure for $O(\log^\ast n)$ rounds suffices to color all vertices.

Similar ideas have also been applied in other papers~\cite{SchneiderW10,ElkinPS15,ChangKP19}.
However, for technical reasons, we cannot directly apply the results in these papers.
The main difficulty in our setting is that we need to deal with \emph{oriented} graphs with widely varying out-degrees, palette sizes, and excess colors;
the guaranteed number of excess colors at a vertex depends on its out-degree,
\emph{not} the global parameter $\Delta$.

Lemma~\ref{lem:color-remain} summarizes the properties of our ultrafast coloring algorithm
when each vertex has many excess colors;
its proof appears in Section~\ref{sect:trial-detail}.
Recall that $\Psi(v)$ denotes the palette of $v$,
so $|\Psi(v)| - \deg(v)$ is the number of excess colors at $v$.

\begin{lemma}\label{lem:color-remain}
Consider a directed acyclic graph, where vertex $v$ is associated with a parameter $p_v \leq |\Psi(v)| -  \deg(v)$
We write $p^\star = \min_{v\in V} p_v$.
Suppose that there is a number $C = \Omega(1)$ such that all vertices $v$  satisfy $\sum_{u \in  N_{\operatorname{out}}(v)} 1/ p_u \leq 1/C$.
Let $d^\star$ be the maximum out-degree of the graph.
There is an algorithm that takes $O\left(1 + \log^\ast p^\star - \log^\ast C\right)$ time and achieves the following.
Each vertex $v$ remains uncolored with probability at most $\exp(-\Omega(\sqrt{p^\star})) + d^\star \exp(-\Omega(p^\star))$.
This is true even if the random bits generated outside a constant radius around $v$ are determined adversarially.
\end{lemma}

We briefly explain the intuition underlying Lemma~\ref{lem:color-remain}.
Consider the following coloring procedure. Each vertex selects $C/2$ colors from its available colors randomly.
Vertex $v$ successfully colors itself if at least one of its selected colors is not in conflict with
any color selected by vertices in $N_{\text{out}}(v)$.
For each color $c$ selected by $v$,
the probability that $c$ is also selected by some vertex in $N_{\text{out}}(v)$ is $ (C/2) \sum_{u \in  N_{\operatorname{out}}(v)} 1/ p_u
\leq 1/2$. Therefore, the probability that $v$ still remains uncolored after this procedure is $\exp(-\Omega(C))$, improving the gap between the number of excess colors and the out-degree (i.e., the parameter $C$) exponentially. We are done after repeating this procedure for
$O(1 + \log^\ast p^\star - \log^\ast C)$ rounds.
Lemma~\ref{lem:color-remain-simple} is a more user-friendly version of Lemma~\ref{lem:color-remain} for simpler situations.

\begin{lemma}\label{lem:color-remain-simple}
Suppose $|\Psi(v)| \geq (1+\rho)\Delta$ for each vertex $v$, and $\rho = \Omega(1)$.
There is an algorithm that takes $O\left(1 + \log^\ast \Delta - \log^\ast \rho\right)$ time and achieves the following.
Each vertex $v$ remains uncolored with probability at most $\exp(-\Omega(\sqrt{\rho \Delta}))$.
This is true even if the random bits generated outside a constant radius around $v$ are determined adversarially.
\end{lemma}

\begin{proof}
We apply Lemma~\ref{lem:color-remain}.
Orient the graph arbitrarily,
and then set $p_v = \rho\Delta$ for each $v$.
Use the parameters $C = \rho$, $p^\star = \rho\Delta$, and $d^\star = \Delta$.
The time complexity is $O\left(1 + \log^\ast p^\star - \log^\ast C\right)= O\left(1 + \log^\ast \Delta - \log^\ast \rho\}\right) $.
The failure probability is $\exp(-\Omega(\sqrt{p^\star})) + d^\star \exp(-\Omega(p^\star)) = \exp(-\Omega(\sqrt{\rho \Delta}))$.
\end{proof}

\subsection{Gaining Excess Colors}

Schneider and Wattenhofer~\cite{SchneiderW10} illustrated that vertex coloring can be performed
very quickly, given enough excess colors.  However,
in the $(\Delta+1)$-list coloring problem there \emph{is just one excess color} initially,
so the problem is how to \emph{create} them.
Elkin, Pettie, and Su~\cite{ElkinPS15} observed that
if the graph induced by $N(v)$ is not too dense, then $v$ can obtain a significant number of excess colors after
\emph{one} iteration of the following simple random coloring routine.
 Each vertex $v$, with probability $1/5$, selects a color $c$ from its palette $\Psi(v)$ uniformly at random;
then vertex $v$ successfully colors itself by $c$ if $c$ is not chosen by any vertex in $N(v)$.
Intuitively, if $N(v)$ is not too close to a clique, then
a significant number of \emph{pairs} of vertices in the neighborhood $N(v)$ get assigned
the same color.  Each such pair effectively reduces $v$'s palette size by 1 but its degree by 2,
thereby increasing the number of excess colors at $v$ by 1.

There are many \emph{global} measures of sparsity,
such as arboricity and degeneracy.
We are aware of two locality-sensitive ways to measure it: the \emph{$(1-\epsilon)$-local sparsity} of~\cite{AlonKS99,ElkinPS15,MolloyR97,Vu02},
and the \emph{$\epsilon$-friends} from~\cite{HarrisSS18},
defined formally as follows.

\begin{definition}[\cite{ElkinPS15}]\label{def-sparse-1}
A vertex $v$ is $(1-\epsilon)$-locally sparse if the subgraph induced by $N(v)$
has at most $(1-\epsilon){\Delta\choose 2}$ edges; otherwise $v$ is $(1-\epsilon)$-locally dense.
\end{definition}

\begin{definition}[\cite{HarrisSS18}]\label{def-sparse-2}
An edge $e=\{u,v\}$ is an $\epsilon$-friend edge if $|N(u) \cap N(v)| \geq (1-\epsilon)\Delta$.
We call $u$ an {$\epsilon$-friend} of $v$ if $\{u,v\}$ is an { $\epsilon$-friend edge}.
A vertex $v$ is {$\epsilon$-dense} if $v$ has at least $(1-\epsilon)\Delta$ $\epsilon$-friends,
otherwise it is {$\epsilon$-sparse}.
\end{definition}

Throughout this paper, we only use Definition~\ref{def-sparse-2}.
Lemma~\ref{lem:initial-color} shows that in $O(1)$ time we can create excess colors at all locally sparse vertices.

\begin{lemma}\label{lem:initial-color}
Consider the $(\Delta+1)$-list coloring problem.
There is an $O(1)$-time algorithm that colors a subset of vertices such that the following is true for each $v \in V$ with
$\deg(v) \geq (5/6)\Delta$.
\begin{enumerate}
\item[(i)] With probability  $1 - \exp(-\Omega(\Delta))$, the number of uncolored neighbors of $v$ is at least $\Delta / 2$.
\item[(ii)] 
    With probability  $1 - \exp(-\Omega(\epsilon^2 \Delta))$,
    $v$ has at least $\Omega(\epsilon^2 \Delta)$ excess colors,
    where $\epsilon$ is the highest value such that $v$ is $\epsilon$-sparse.
\end{enumerate}
\end{lemma}

The algorithm behind Lemma~\ref{lem:initial-color} is the random coloring routine described above.
If a vertex $v$ is $\epsilon$-sparse, then there must be $\Omega(\epsilon^2 \Delta^2)$ pairs of vertices $\{u,w\} \subseteq N(v)$ such that $\{u,w\}$ is \underline{\emph{not}} an edge.
If $|\Psi(u) \cap \Psi(w)| = \Omega(\Delta)$,\footnote{If the condition is not met, then we have  $| (\Psi(u) \cup \Psi(w)) \setminus \Psi(v) | = \Omega(\Delta)$, and so with constant probability one of $u$ and $v$ successfully colors itself with a color not in $\Psi(v)$, and this also increases the number of excess colors at $v$ by $1$.} then
the probability that both $u$ and $v$ are colored by the same color is $\Omega(1/\Delta)$,  so the expected number of excess colors created at $v$ is at least  $\Omega\left(\frac{\epsilon^2 \Delta^2}{\Delta}\right) = \Omega(\epsilon^2 \Delta)$.

Similar but slightly weaker lemmas were proved in~\cite{ElkinPS15,HarrisSS18}.  The corresponding lemma from~\cite{ElkinPS15}
does not apply to \emph{list} coloring, and the corresponding lemma from~\cite{HarrisSS18}
obtains a high probability bound only if $\epsilon^4 \Delta = \Omega(\log n)$.
Optimizing this requirement is of importance, since this is the threshold about how locally
sparse a vertex needs to be in order to obtain excess colors.
Since this is not the main contribution of this work, the proof of Lemma~\ref{lem:initial-color} appears in Appendix~\ref{sect:oneshot-detail}.

The notion of local sparsity
is especially useful for addressing the $(2\Delta-1)$-edge coloring problem~\cite{ElkinPS15}, since it can be phrased as $(\Delta'+1)$-vertex
coloring the \emph{line graph} ($\Delta'= 2\Delta-2$), which is everywhere $(\frac12+o(1))$-locally sparse and is also everywhere $(\frac12 - o(1))$-sparse.

\subsection{Coloring Locally Dense Vertices}
In the vertex coloring problem we cannot count on any kind of local sparsity, so the next challenge
is to make local \emph{density} also work to our advantage.
Harris, Schneider, and Su~\cite{HarrisSS18} developed a remarkable new graph decomposition that can be computed in $O(1)$
rounds of communication.  The decomposition takes a parameter $\epsilon$, and partitions the vertices
into an $\epsilon$-{sparse} set, and several vertex-disjoint $\epsilon$-\emph{dense} components induced by the $\epsilon$-friend edges,
each with weak diameter at most 2.

Based on this decomposition, they designed a $(\Delta+1)$-list coloring algorithm that takes $O(\sqrt{\log \Delta} + \log\log n + \Detd(\polylog n)) = O(\sqrt{\log \Delta}) + 2^{O(\sqrt{\log \log n})}$ time. We briefly overview their algorithm, as follows.

\paragraph{Coloring $\epsilon$-Sparse Vertices.} By utilizing the excess colors, Harris et al.~\cite{HarrisSS18} showed that the $\epsilon$-sparse set can be colored in $O(\log\epsilon^{-1} + \log\log n + \Detd(\polylog n))$
time using techniques in~\cite{ElkinPS15} and~\cite{BEPS16}. More specifically, they applied the algorithm of~\cite[Corollary~4.1]{ElkinPS15} using the $\epsilon' \Delta = \Omega(\epsilon^2 \Delta)$ excess colors, i.e., $\epsilon' = \Theta(\epsilon^2)$. This takes $O\left(\log (\epsilon^{-1})\right) + T\left(n, O\left(\frac{\log^2 n}{\epsilon'}\right)\right)$ time, where $T(n', \Delta') = O(\log \Delta' + \log\log n' + \Detd(\polylog n'))$ is the time complexity of the $(\deg+1)$-list coloring algorithm of~\cite[Theorem~5.1]{BEPS16} on $n'$-vertex graphs of maximum degree $\Delta'$.

\paragraph{Coloring $\epsilon$-Dense Vertices.}
For $\epsilon$-dense vertices,
Harris et al.~\cite{HarrisSS18} proved that by
coordinating the coloring decisions within each
dense component, it takes only $O(\log_{1/\epsilon} \Delta + \log\log n + \Detd(\polylog n))$ time to color the dense sets,
i.e., the bound \emph{improves} as $\epsilon\rightarrow 0$.
The time for the overall algorithm is minimized by choosing $\epsilon = \exp(-\Theta(\sqrt{\log\Delta}))$.

The algorithm for coloring $\epsilon$-dense vertices first applies $O(\log_{1/\epsilon} \Delta)$ iterations of {\em dense coloring steps} to reduce the maximum degree to $\Delta' = O(\log n) \cdot 2^{O\left(\log_{1/\epsilon} \Delta\right)}$, and then apply the $(\deg+1)$-list coloring algorithm of~\cite[Theorem~5.1]{BEPS16} to color the remaining vertices in $O(\log \Delta' + \log\log n + \Detd(\polylog n)) = O(\log_{1/\epsilon} \Delta + \log\log n + \Detd(\polylog n))$ time.

In what follows, we informally sketch the idea behind the dense coloring steps.
To finish in $O(\log_{1/\epsilon} \Delta)$ iterations, it suffices that the maximum degree is reduced by a factor of $\epsilon^{-\Omega(1)}$ in each iteration. Consider an $\epsilon$-dense vertex $v$ in a component $S$ induced by the $\epsilon$-friend edges.
Harris et al.~\cite{HarrisSS18} proved that the number of $\epsilon$-dense neighbors of $v$ that are not in $S$ is at most
$\epsilon \Delta$.
Intuitively, if we let each dense component output a random coloring that has no conflict within the component, then the probability that the color choice of a vertex $v \in S$ is in conflict with an {\em external neighbor} of $v$ is $O(\epsilon)$.
Harris et al.~\cite{HarrisSS18} showed that this intuition can be nearly realized, and they developed a coloring procedure that is able to reduce the maximum degree by a factor of $O(\sqrt{\epsilon^{-1}})$
in each iteration.

\subsection{New Results}
In this paper we give a fast randomized algorithm for $(\Delta+1)$-vertex coloring.  It is based on a 
hierarchical
version of the Harris-Schneider-Su decomposition with
$\log\log\Delta-O(1)$ levels determined by an increasing
sequence of sparsity thresholds $(\epsilon_1,\ldots,\epsilon_\ell)$, with $\epsilon_i = \sqrt{\epsilon_{i+1}}$.
Following~\cite{HarrisSS18}, we
begin with a single iteration of the {\em initial coloring step} (Lemma~\ref{lem:initial-color}), in which a constant fraction of the vertices are colored.  The
guarantee of this procedure is that any vertex $v$ at the $i$th layer (which is $\epsilon_i$-dense but $\epsilon_{i-1}$-sparse),
has $\Omega(\epsilon_{i-1}^2\Delta)$ \emph{pairs} of vertices in its neighborhood $N(v)$ assigned the same color,
thereby creating that many excess colors in the palette of $v$.

At this point, the most natural way to proceed is to apply a Harris-Schneider-Su style {\em dense coloring step}
to each layer, with the hope that each will take roughly constant time.
Recall that (i) any vertex $v$ at the $i$th layer already has $\Omega(\epsilon_{i-1}^2\Delta)$ excess colors, and (ii) the dense coloring step reduces the maximum degree by a factor of $\epsilon^{-\Omega(1)}$ in each iteration.
Thus, in $O\left(\log_{1/\epsilon_i} \frac{\Delta}{\epsilon_{i-1}^{2.5} \Delta}\right)=O(1)$ time we should be able to
create a situation where any uncolored vertices have $O(\epsilon_{i-1}^{2.5}\Delta)$ uncolored neighbors but
$\Omega(\epsilon_{i-1}^2\Delta)$ excess colors in their palette.  With such a large gap,
a Schneider-Wattenhofer style coloring algorithm (Lemma~\ref{lem:color-remain-simple}) should complete in very few additional steps.

It turns out that in order to color $\epsilon_i$-dense components efficiently, we need to maintain relatively large
\emph{lower bounds} on the available palette and relatively small \emph{upper bounds} on the number
of external neighbors (i.e., the neighbors outside the $\epsilon_i$-dense component).  Thus, it is important that when we first
consider a vertex, we have not already colored too many of its neighbors.
Roughly speaking, our algorithm classifies the dense blocks
at layer $i$ into \emph{small}, \emph{medium}, and \emph{large}
based chiefly on the block size,
and partitions the set of all
blocks of all layers into $O(1)$ {\em groups}.
We apply the dense coloring steps {\em in parallel} for all blocks in the same group.
Whenever we process a block $B$, we need to make sure that all its vertices have a large enough  palette.
For large blocks, the palette size guarantee comes from the lower bound on the block size.
For small and medium blocks, the palette size guarantee comes from the ordering of the blocks being processed; we will show that whenever a small or medium block $B$ is considered, each vertex $v \in B$ has a sufficiently large number of neighbors that have yet to be colored.

All of the coloring steps outlined above finish in
$O(\log^\ast \Delta)$ time.
The  bottleneck procedure is the algorithm of Lemma~\ref{lem:color-remain-simple}, and the rest takes only $O(1)$ time.
Each of these coloring steps may not color all vertices it considers.
The vertices left uncolored are put in $O(1)$ classes,
each of which either induces a bounded degree graph
or is composed of $O(\polylog n)$-size components, w.h.p.
The former type can be colored deterministically
in $O(\log^\ast n)$ time and the latter
in $\Detd(\polylog n)$ time.
In view of Linial's lower bound~\cite{Linial92}
we have $\Detd(\polylog n) = \Omega(\log^\ast n)$ and
the running time of our $(\Delta+1)$-list coloring algorithm
is
\[
O(\log^\ast \Delta) + O(\log^\ast n) + O(\Detd(\polylog n)) = O(\Detd(\polylog n)).
\]

\paragraph{Recent Developments.}
After the initial publication of this work~\cite{ChangLP18},
our algorithm was adapted to solve $(\Delta+1)$-coloring
in several other models of computation,
namely the \emph{congested clique},
the \emph{MPC}\footnote{massively parallel computation} model,
and the \emph{centralized local computation} model~\cite{AssadiCK18,Parter18,ParterS18,ChangFGUZ18}.
Chang, Fischer, Ghaffari, Uitto, and Zheng~\cite{ChangFGUZ18},
improving~\cite{Parter18,ParterS18}, showed
that $(\Delta+1)$-coloring can be solved in the congested
clique in $O(1)$ rounds, w.h.p.
In the MPC model,
Assadi, Chen, and Khanna~\cite{AssadiCK18} solve
$(\Delta+1)$-coloring in $O(1)$ rounds using $\tilde{O}(n)$ memory per machine, whereas Chang et al.~\cite{ChangFGUZ18} solve
it in $O(\sqrt{\log\log n})$ time with just $O(n^\epsilon)$ memory per machine.  In the centralized local computation model,
Chang et al.~\cite{ChangFGUZ18} proved that $(\Delta+1)$-coloring queries can be answered with just \emph{polynomial}
probe complexity $\Delta^{O(1)}\log n$.

\paragraph{Organization.}
In Section~\ref{sect:hierarchy} we define a hierarchical decomposition based on~\cite{HarrisSS18}.
Section~\ref{sect:algo} gives a high-level description of the algorithm, which uses a variety of coloring routines
whose guarantees are specified by the following lemmas.
\begin{itemize}
\item Lemma~\ref{lem:color-remain} analyzes
the procedure \trial, which is a generalization of the Schneider-Wattenhofer coloring routing; it is proved in Section~\ref{sect:trial-detail}.
\item Lemma~\ref{lem:initial-color}
shows that the procedure \oneshot{} creates many excess colors; it is proved in Appendix~\ref{sect:oneshot-detail}.
\item Lemmas~\ref{lem:color-sm}--\ref{lem:color-top-lg} analyze two versions of an algorithm \dense,
which is a generalization of the Harris-Schneider-Su routine~\cite{HarrisSS18} for coloring locally dense vertices; they are proved in Section~\ref{sect:dense}.
\end{itemize}
Appendix~\ref{sect:tools} reviews all of the standard concentration inequalities that we use.


\section{Hierarchical Decomposition}\label{sect:hierarchy}

In this section, we extend the work of Harris, Schneider, and Su~\cite{HarrisSS18} to
define a hierarchical decomposition of the vertices based on local sparsity.
Let $G=(V,E)$ be the input graph, $\Delta$ be the maximum degree, and $\epsilon \in (0,1)$ be a parameter.
An edge $e=\{u,v\}$ is an {\em $\epsilon$-friend edge} if $|N(u) \cap N(v)| \geq (1-\epsilon)\Delta$.
We call $u$ an {\em $\epsilon$-friend} of $v$ if $\{u,v\}$ is an {\em $\epsilon$-friend edge}.
A vertex $v$ is called {\em $\epsilon$-dense} if $v$ has at least $(1-\epsilon)\Delta$ $\epsilon$-friends,
otherwise it is {\em $\epsilon$-sparse}.  Observe that it takes one round of communication to tell whether
each edge is an $\epsilon$-friend, and hence one round 
for each vertex to decide if it is 
$\epsilon$-sparse or $\epsilon$-dense.

We write $V_{\epsilon}^{\sparseexp}$ (and $V_{\epsilon}^{\denseexp}$) to be the set of $\epsilon$-sparse (and $\epsilon$-dense) vertices.
Let $v$ be a vertex in a set $S\subseteq V$ and $V'\subseteq V$.
Define $\bar{d}_{S,V'}(v) = |(N(v)\cap V') \setminus S|$ to be the {\em external degree} of $v$
with respect to $S$ and $V'$, and $a_S(v) = |S \setminus (N(v) \cup \{v\})|$ to be the {\em anti-degree} of $v$ with respect to $S$.
A connected component $C$ of the subgraph formed by the $\epsilon$-dense 
vertices and the $\epsilon$-friend edges is called an 
{\em $\epsilon$-almost clique}.
This term makes sense in the context of Lemma~\ref{lem:cluster-property} from~\cite{HarrisSS18},
which summarizes key properties of almost cliques.

\begin{lemma}[\cite{HarrisSS18}]\label{lem:cluster-property}
Fix any $\epsilon < 1/5$.
The following conditions are met for each $\epsilon$-almost clique $C$, and each vertex $v \in C$.
\begin{enumerate}
\item[(i)] $\bar{d}_{C,V_{\epsilon}^{\denseexp}}(v) \leq \epsilon \Delta$. (Small external degree w.r.t.~$\epsilon$-dense vertices.)
\item[(ii)] $a_C(v) < 3\epsilon \Delta$.  (Small anti-degree.)
\item[(iii)] $|C| \leq (1+3\epsilon)\Delta$.  (Small size, a consequence of (ii).)
\item[(iv)] $\dist_G(u,v) \leq 2$ for each $u,v \in C$.  (Small weak diameter.)
\end{enumerate}
\end{lemma}

Lemma~\ref{lem:cluster-property}(iv) implies that \emph{any} sequential algorithm operating solely 
on $C$ can be simulated in $O(1)$ rounds in the $\LOCAL$ model.  The node in $C$ with minimum ID
can gather all the relevant information from $C$ in 2 rounds of communication, compute the output of the algorithm locally, and disseminate these results in another 2 rounds of communication.  For example, the \dense{} algorithm (versions 1 and 2) presented in Section~\ref{sect:dense} are nominally sequential algorithms but can be implemented in $O(1)$ distributed rounds.

\subsection{A Hierarchy of Almost Cliques}

Throughout this section, we fix
some increasing sequence of sparsity parameters $(\epsilon_1, \ldots, \epsilon_{\ell})$
and a subset of vertices $V^\star \subseteq V$, which, roughly speaking, are those left uncolored by the initial coloring procedure of Lemma~\ref{lem:initial-color} 
and also satisfy the two conclusions of Lemma~\ref{lem:initial-color}(i,ii).
The sequence $(\epsilon_1, \ldots, \epsilon_{\ell})$ always adheres to Definition~\ref{def:sparsity}.

\begin{definition}\label{def:sparsity}
A sequence $(\epsilon_1, \ldots, \epsilon_{\ell})$ is a valid \emph{sparsity sequence} if the following conditions are met:
(i) $\epsilon_i = \sqrt{\epsilon_{i-1}} = \epsilon_1^{2^{-(i-1)}}$,
 and
(ii) $\epsilon_\ell \leq 1/K$ for some sufficiently large $K$.
\end{definition}

\paragraph{Layers.} Define $V_1 = V^\star \cap V_{\epsilon_1}^{\denseexp}$
and $V_i = V^\star \cap (V_{\epsilon_i}^{\denseexp} \setminus V_{\epsilon_{i-1}}^{\denseexp})$, for $i > 1$.
Define $V_{\sparse} = V^\star \cap V_{\epsilon_\ell}^{\sparseexp} = V^\star \setminus (V_1 \cup \cdots \cup V_{\ell})$.
It is clear that $(V_1,\ldots, V_{\ell}, V_{\sparse})$ is a partition of $V^\star$.
We call $V_i$ the {\em layer-$i$} vertices, and call $V_{\sparse}$ the {\em sparse vertices}.
In other words, $V_i$ is the subset of $V^\star$ that are $\epsilon_i$-dense but  $\epsilon_{i-1}$-sparse.
Remember that the definition of sparsity is with respect to the entire graph $G=(V,E)$ not
the subgraph induced by $V^\star$.

\paragraph{Blocks.} The layer-$i$ vertices $V_i$ are partitioned into {\em blocks} as follows.
Let $\{C_1, C_2,\ldots\}$ be the set of $\epsilon_i$-almost cliques, and let
$B_j = C_j \cap V_i$.
Then $(B_1,B_2,\ldots)$ is a partition of $V_i$.
Each $B_j \neq \emptyset$ is called a \emph{layer-$i$ block}.
See Figure~\ref{fig:blocks} for an illustration; 
the shaded region indicates a layer-$i$ block $B$
and the hollow regions are those $\epsilon_{i-1}$-almost 
cliques.

A layer-$i$ block $B$ is a {\em descendant} of a layer-$i'$ block $B'$, $i<i'$,
if $B$ and $B'$ are both subsets of the same $\epsilon_{i'}$-almost clique.
Therefore, the set of all blocks in all layers naturally forms a rooted tree $\mathcal{T}$, where the root represents
$V_{\sparse}$, and every other node represents a block in some layer.
For example, in Figure~\ref{fig:blocks}, the blocks contained in 
$C_1,\ldots,C_k$ are at layers $1,\ldots,i-1$, and are all descendants
of $B$.

\begin{figure}
\begin{center}
\includegraphics[width=.6\linewidth]{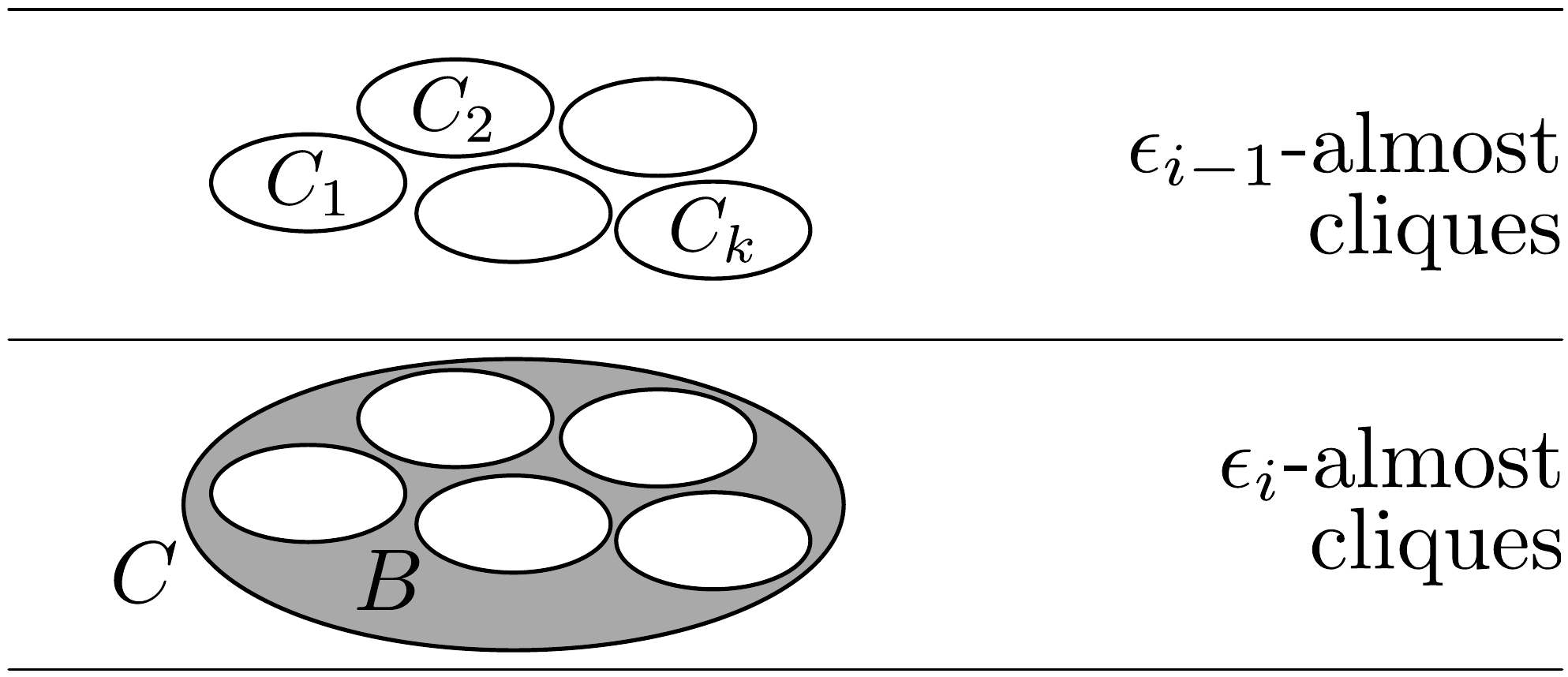}
\end{center}
\caption{Almost-cliques and blocks.}
\label{fig:blocks}
\end{figure}

\subsection{Block Sizes and Excess Colors}

We classify the blocks into three types: {\em small}, {\em medium}, and {\em large}. 
A block $B$ at layer $i$ is called \emph{large-eligible} if 
\[
|B| \ge \frac{\Delta}{\log (1/\epsilon_i)}.
\]

\begin{description}
\item[Large blocks.] The set of large blocks is a maximal set of unrelated, large-eligible
blocks, which prioritizes blocks by size, breaking ties by layer.  More formally, a large-eligible layer-$i$ block $B$ is large if and only if, for every large-eligible $B'$ at layer $j$ that is an ancestor or descendant of $B$, either $|B'| < |B|$ or $|B'|=|B|$ and $j<i$.

\item[Medium blocks.]  Every large-eligible block that is not large is a medium block.

\item[Small blocks.] All other blocks are small.
\end{description}

Define $V_i^{\SM}$,  $V_i^{\MD}$,  and $V_i^{\LG}$, to be, respectively,
the sets of all vertices in layer-$i$ small blocks, layer-$i$ medium blocks, and layer-$i$ large blocks.
For each $X \in \{\SM,\MD,\LG\}$, we write $V_{2+}^{X} = \bigcup_{i=2}^{\ell} V_i^X$ 
to be the set of all vertices of type $X$, excluding those in layer 1.

\paragraph{Overview of Our Algorithm.} The decomposition and $\mathcal{T}$ are trivially computed in $O(1)$ rounds of communication.  
The first step of our algorithm 
is to execute an $O(1)$-round coloring procedure (\oneshot)
which colors a small constant fraction of the vertices in $G$; the relevant guarantees
of this algorithm were stated in Lemma~\ref{lem:initial-color}.
Let $V^\star$ be the subset of uncolored vertices that, in addition, satisfy the conclusions of Lemma~\ref{lem:initial-color}(i,ii).
Once $V^\star$ is known, it can be partitioned into the following sets
\[
\left(V_1^{\SM}, \ldots, V_{\ell}^{\SM}, V_1^{\MD}, \ldots, V_{\ell}^{\MD}, V_1^{\LG}, \ldots, V_{\ell}^{\LG}, V_{\sparse}\right)
\]
These are determined by the hierarchical decomposition with respect to a particular 
sparsity sequence $(\epsilon_1, \ldots, \epsilon_\ell)$.\footnote{Note that the 
classification of vertices into small, medium, and large blocks can only 
be done \emph{after} \oneshot{} is complete.  Recall that if $C$ is an $\epsilon_i$-almost clique, $B = C\cap V_i$ is the subset of $C$ that is both
$\epsilon_{i-1}$-sparse \underline{and} uncolored by \oneshot.  Thus, whether
the layer-$i$ block in $C$ is large-eligible depends on how many vertices are 
successfully colored.}
We color the vertices of $V^\star \setminus V_{\sparse}$ in six stages
according to the ordering 
\[
\left(V_{2+}^{\SM}, V_{1}^{\SM}, V_{2+}^{\MD}, V_{1}^{\MD}, V_{2+}^{\LG}, V_{1}^{\LG}\right).
\]
As we argue below, coloring vertices in the order small, medium, large ensures
that when a vertex is considered, it has \emph{sufficiently many} remaining colors in its palette,
as formalized by Lemma~\ref{lem:palette-lb} below.
The reason for dealing with layer-1 vertices separately
stems from the fact that a vertex at layer $i>1$ is known to be $\epsilon_i$-dense
but $\epsilon_{i-1}$-sparse, but layer-1 vertices are not known to have 
any non-trivial sparsity. 
At the end of this process a small portion of vertices $U \subseteq V^\star \setminus V_{\sparse}$ may remain uncolored.
However, they all have sufficiently large palettes such that $U \cup V_{\sparse}$ 
can be colored efficiently in $O(\log^\ast n)$ time.

\begin{lemma}\label{lem:palette-lb}
For each layer $i \in [1,\ell]$, the following is true.
\begin{itemize}
\item For each $v \in V_i^{\SM}$ with $|N(v) \cap V^\star| \geq \Delta/3$, we have
$|N(v) \cap (V_{2+}^{\MD} \cup V_{1}^{\MD} \cup V_{2+}^{\LG} \cup V_{1}^{\LG} \cup V_{\sparse})| \geq \frac{\Delta}{4}$.
\item  For each $v \in V_i^{\MD}$, we have
$|N(v) \cap (V_{2+}^{\LG} \cup V_{1}^{\LG} \cup V_{\sparse})| \geq \frac{\Delta}{2\log (1/\epsilon_i)}$.
\end{itemize}
\end{lemma}

In other words, regardless of how we proceed to \emph{partially} color the vertices in small blocks,
each $v \in V_i^{\SM}$ always has at least $\frac{\Delta}{4}$ available colors in its palette, 
due to the number of its (still uncolored) neighbors in medium and large blocks, and $V_{\sparse}$.
Similarly, regardless of how we partially color the vertices in small and medium blocks,
each $v \in V_i^{\MD}$ always has at least $\frac{\Delta}{2\log (1/\epsilon_i)}$ available colors in its palette.

Before proving Lemma~\ref{lem:palette-lb} we first 
establish a useful property that constrains the structure of the block hierarchy $\mathcal{T}$.
Intuitively, Lemma~\ref{lem:children-in-T} shows that a node (block) in $\mathcal{T}$
can have exactly one child of essentially any size, but if it has two or more children
then the union of all strict descendants must be very small.

\begin{lemma}\label{lem:children-in-T}
Let $C$ be an $\epsilon_i$-almost clique and $C_1,\ldots,C_l$ be the
$\epsilon_{i-1}$-almost cliques contained in $C$.  Either $l=1$
or $\sum_{j=1}^l |C_j| \le 2(3\epsilon_i + \epsilon_{i-1})\Delta$.
In particular, if $B$ is the layer-$i$ block contained in $C$,
either $B$ has one child in $\mathcal{T}$ or the number of vertices
in all strict descendants of $B$ is at most 
$2(3\epsilon_i + \epsilon_{i-1})\Delta < 7\epsilon_i\Delta$.
\end{lemma}
\begin{proof}
Suppose, for the purpose of obtaining a contradiction, that
$l\ge 2$ and $\sum_{j=1}^l |C_j| > 2(3\epsilon_i + \epsilon_{i-1})\Delta$.
Without loss of generality, suppose $C_1$ is the smallest, so
$\sum_{j=2}^l |C_j| > (3\epsilon_i + \epsilon_{i-1})\Delta$.
Any $v\in C_1$ is $\epsilon_{i-1}$-dense and therefore has at
least $(1-\epsilon_{i-1})\Delta$ neighbors that are $\epsilon_{i-1}$-friends.
By the anti-degree property of Lemma~\ref{lem:cluster-property},
$v$ is adjacent to all but at most $3\epsilon_i\Delta$ vertices in $C$.
Thus, by the pigeonhole principle 
$v$ is joined by edges to more than $\epsilon_{i-1}\Delta$ 
members of $C_2 \cup\cdots\cup C_l$.
By the pigeonhole principle again, at least one of these edges is
one of the $\epsilon_{i-1}$-friend edges incident to $v$.
This means that $C_1$ \emph{cannot} be a connected component in 
the graph formed by $\epsilon_{i-1}$-dense vertices and $\epsilon_{i-1}$-friend edges.
\end{proof}

\begin{proof}[Proof of Lemma~\ref{lem:palette-lb}]
First consider the case of $v \in V_i^{\MD}$.
Let $B$ be the layer-$i$ medium block containing $v$.
Every medium block is large-eligible but not large, meaning it must have a large 
ancestor or descendant $B'$ with at least as many vertices.  
If $B'$ is a layer-$j$ block, then
\[
|B'| \:=\: \max\left\{\left|B'\right|,\, |B|\right\} \:\geq\: \frac{\Delta}{\log (1/\epsilon_k)}, \ \mbox{ where $k = \max\{i,j\}$}.
\]
Let $C$ be the layer-$k$ almost clique containing both $B$ and $B'$.
By Lemma~\ref{lem:cluster-property}, $v$ has at most 
$3\epsilon_k\Delta$ non-neighbors in $C$, which, since $B' \subseteq C$, means
that the number of neighbors of $v$ in $B'$ is at least
\begin{align*}
|B'| - 3\epsilon_{k}\Delta
 &\ge \frac{\Delta}{\log (1/\epsilon_k)} - 3\epsilon_{k}\Delta &\\
 &\ge \frac{\Delta}{2\log (1/\epsilon_k)}                 & \{\epsilon_k \le \epsilon_\ell \mbox{ sufficiently small}\}\\
 &\ge \frac{\Delta}{2\log (1/\epsilon_i)}                 & \{\log(1/\epsilon_k) \le \log(1/\epsilon_i)\}
\end{align*}
Therefore, $|N(v) \cap (V_{2+}^{\LG} \cup V_{1}^{\LG} \cup V_{\sparse})| \geq \frac{\Delta}{2\log (1/\epsilon_i)}$.

Now consider any vertex $v \in V_i^{\SM}$ with $|N(v)\cap V^\star| \ge \Delta/3$.
Let $B$ be the layer-$i$ small block containing $v$.
We partition the set $N(v)\cap V^\star$ into three groups 
$A_1 \cup A_2 \cup A_3$.
\begin{align*}
A_1 &= N(v) \cap \left(V_{2+}^{\MD} \cup V_{1}^{\MD} \cup V_{2+}^{\LG} \cup V_{1}^{\LG} \cup V_{\sparse}\right).\istrut[2]{0}\\
A_2 &= \mbox{the neighbors in all ancestor and descendant small blocks of $B$, including $B$.}\istrut[2]{0}\\
A_3 &= \mbox{the remaining neighbors.}
\end{align*}
To prove the lemma, it suffices to show that $|A_1| \geq \frac{\Delta}{4}$.
Since $|A_1\cup A_2\cup A_3| \ge \frac{\Delta}{3}$, we need 
to prove $|A_2\cup A_3| \leq \frac{\Delta}{12}$.
We first bound $|A_3|$, then $|A_2|$.

Note that $v$ is $\epsilon_j$-dense for $j\in [i,\ell]$,
so, according to Lemma~\ref{lem:cluster-property},
$v$ must have at least $(1-\epsilon_j)\Delta$ $\epsilon_j$-friends.
Let $u$ be any neighbor of $v$ not in an ancestor/descendant 
of $B$, which means that either 
(i) $u\in V_{\sparse}$ or 
(ii) for some $j\in [i,\ell]$, $v$ and $u$
are in distinct $\epsilon_j$-almost cliques.
In case (i) $u$ is counted in $A_1$.  In case (ii)
it follows that $u$ \emph{cannot} be an $\epsilon_j$-friend
of $v$.  Since, by Lemma~\ref{lem:cluster-property},
$v$ has at most $\epsilon_j\Delta$ $\epsilon_j$-non-friends,
\[
|A_3| \le \sum_{j=i}^{\ell} \epsilon_{j} \Delta < 2\epsilon_{\ell}\Delta.
\]
We now turn to $A_2$.  Define $i^\star \in [1,i-1]$ to be the largest
index such that $B$ has at least two descendants at layer $i^\star$, or
let $i^\star=0$ if no such index exists.
Let $A_{2,\text{low}}$ be the set of vertices in $A_2$ residing in blocks at 
layers $1, \ldots, i^\star$, and let $A_{2,\text{high}} = A_2 \setminus A_{2,\text{low}}$.
By the definition of small blocks,
\begin{align*}
|A_{2,\text{high}}| 
 &< \sum_{j= i^\star + 1}^\ell \frac{\Delta}{\log (1/\epsilon_j)}\\
 &< \frac{2\Delta}{\log (1/\epsilon_{\ell})}.
    & \{\mbox{geometric sum}\}
\end{align*}
If $i^\star=0$ then $A_{2,\text{low}}=\emptyset$. Otherwise, 
by Lemma~\ref{lem:children-in-T}, the number of vertices in $A_{2,\text{low}}$ is at most $7\epsilon_{i^\star + 1}\Delta \leq 7\epsilon_i \Delta \leq 7\epsilon_\ell \Delta$.
Since $\epsilon_\ell$ is a sufficiently small constant,
\[
|A_2 \cup A_3| < 2\epsilon_\ell\Delta + \frac{2\Delta}{\log (1/\epsilon_{\ell})} + 7\epsilon_\ell \Delta < \Delta/12,
\]
which completes the proof.
\end{proof}

\begin{remark}
In the preliminary version of this paper~\cite{ChangLP18}, the algorithm for coloring locally dense vertices consisted of $O(\log^\ast \Delta)$ stages. In this paper we improve the number of stages to $O(1)$.
This improvement does not affect the overall asymptotic time in the $\LOCAL$ model
but it simplifies the algorithm, and is critical to \emph{adaptations}
of our algorithm to models in which 
Linial's lower bound~\cite{Linial92} does not apply, 
e.g., the congested clique~\cite{AssadiCK18,ChangFGUZ18}.
\end{remark}

\begin{remark}
The reader might wonder why the definition of medium blocks is needed, as all layer-$i$ medium blocks already have the block size lower bound $\frac{\Delta}{\log (1/\epsilon_j)}$, which guarantees a sufficiently large palette size lower bound for the vertices therein.
It might be possible to consider all the medium blocks as large blocks, but this will destroy the property that for any two blocks $B$ and $B'$ in different layers, if $B$ is a descendant of $B'$, then $B$ and $B'$ cannot both be large; without this property, the coloring algorithm for large blocks will likely be more complicated.
\end{remark}

\section{Main Algorithm}\label{sect:algo}

Our algorithm follows the {\em graph shattering} framework~\cite{BEPS16}.
In each step of the algorithm, we specify an invariant that all vertices must satisfy in order to continue to participate.
Those {\em bad vertices} that violate the invariant are removed from consideration; they form connected components of size
$O(\polylog n)$ w.h.p., so we can color them later in $\Detd(\polylog n)$ time.
More precisely, the emergence of the small components is due to the following lemma~\cite{BEPS16,FischerG17}.

\begin{lemma}[The Shattering Lemma] \label{lem:shatter}
Consider a randomized procedure that generates a subset of vertices $B \subseteq V$.
Suppose that for each $v \in V$, we have $\Prob[v \in B] \leq \Delta^{-3c}$, and this holds even if the random bits not in
$N^{c}(v)$ are determined adversarially.
 With probability at least $1 - n^{- \Omega(c')}$, each connected component in the graph induced by
$B$ has size at most $(c'/c) \Delta^{2c} \log_{\Delta} n$.
\end{lemma}

As we will see, some parts of our randomized algorithm consist of $t = O(\log^\ast \Delta)$ steps, and so whether a vertex $v$ is bad actually depends on the random bits in its radius-$t$ neighborhood. Nonetheless, we are still able to apply Lemma~\ref{lem:shatter}. The reason is that we
are able to show that, for any specified constant $k$, each vertex $v$ becomes bad in one particular step with probability at most $\Delta^{-k}$, and this is true regardless of the outcomes in all previous steps and the choices of random bits outside of a constant radius of $v$.

\paragraph{Sparsity Sequence.}
The sparsity sequence for our algorithm is defined by
$\epsilon_1 = \Delta^{-1/10}$,
$\epsilon_i = \sqrt{\epsilon_{i-1}}$ for $i>1$, and
$\ell = \log \log \Delta - O(1)$ is the largest index such that $\frac{1}{\epsilon_{\ell}} \geq  K$ for some sufficiently large constant $K$.

\subsection{Initial Coloring Step}
At any point in time, the number of {\em excess colors} at $v$ is the size of $v$'s remaining palette minus the
number of $v$'s uncolored neighbors.  This quantity is obviously non-decreasing over time.
In the first step of our coloring algorithm, we execute the algorithm of Lemma~\ref{lem:initial-color},
which in $O(1)$ time colors a portion of the vertices.
This algorithm has the property that each remaining
uncolored vertex gains a certain number of excess colors, 
which depends on its local sparsity.
In order to proceed a vertex must satisfy both conditions.
\begin{itemize}
\item If $v$ is $\epsilon_\ell$-dense,  the number of uncolored neighbors of $v$ is at least $\Delta/2$.
\item if $v$ is $\epsilon_i$-sparse, $v$ must have $\Omega(\epsilon_{i}^2 \Delta)$ excess colors.
\end{itemize}

If either condition fails to hold, $v$ is put in the set $V_{{\bad}}$.
We invoke the conditions of Lemma~\ref{lem:initial-color} only with $\epsilon \geq \epsilon_1 = \Delta^{-1/10}$.
Thus, if $\Delta = \Omega(\log^2 n)$, then with high probability (i.e., $1 - 1/\poly(n)$),  $V_{{\bad}} = \emptyset$.
Otherwise, each component of $V_{{\bad}}$ must, by Lemma~\ref{lem:shatter}, have size $O(\poly(\Delta) \cdot \log n) = O(\polylog n)$, w.h.p.
We do not invoke a deterministic algorithm to color $V_{{\bad}}$ just yet.
In subsequent steps of the algorithm, we will continue to add bad vertices to $V_{{\bad}}$.
These vertices will be colored at the end of the algorithm.

\subsection{Coloring Vertices by Layer}

By definition, $V^\star$ is the set of all vertices that
remain uncolored after the initial coloring step \emph{and}
are not put in $V_{{\bad}}$.
The partition 
$V^\star = V_{2+}^{\SM} \cup V_{1}^{\SM} \cup  V_{2+}^{\MD} \cup V_{1}^{\MD} \cup V_{2+}^{\LG} \cup V_{1}^{\LG} \cup V_{\sparse}$
is computed in $O(1)$ time.
In this section, we show how we can color \emph{most} 
of the vertices in 
$V_{2+}^{\SM} \cup V_{1}^{\SM} \cup  V_{2+}^{\MD} \cup V_{1}^{\MD} \cup V_{2+}^{\LG} \cup V_{1}^{\LG}$, in that order,
leaving a small portion of uncolored vertices.

Consider the moment we begin to color $V_{2+}^{\SM}$.  We claim
that each layer-$i$ vertex $v \in V_{2+}^{\SM}$ must have at least 
$\Delta / 6 > \frac{\Delta}{2\log (1/\epsilon_i)}$
\emph{excess colors w.r.t.~$V_{2+}^{\SM}$}. That is, its palette size minus the number of its neighbors in $V_{2+}^{\SM}$ is large.  There are two relevant cases to consider.
\begin{itemize}
\item If the condition  $|N(v) \cap V^\star| \geq \Delta/3$ in Lemma~\ref{lem:palette-lb} is already met, then $v$ has at least $\Delta / 4 > \Delta / 6$ excess colors w.r.t. $V_{2+}^{\SM}$.
\item  Suppose $|N(v) \cap V^\star| < \Delta/3$.
One criterion for adding $v$ to $V_{{\bad}}$ is that $v$ 
is $\epsilon_\ell$-dense but has less than $\Delta/2$ uncolored
neighbors after the initial coloring step.  We know $v$ is $\epsilon_\ell$-dense
and not in $V_{{\bad}}$ (because it is in $V_{2+}^{\SM}$), so it must
have had at least $\Delta/2$ uncolored neighbors after initial coloring.
If $|N(v) \cap V^\star| < \Delta/3$ then at least $(\Delta/2 - \Delta/3)=\Delta/6$
of $v$'s uncolored neighbors must have joined $V_{{\bad}}$, which provide
$v$ with $\Delta/6$ excess colors w.r.t.~$V_{2+}^{\SM}$.
\end{itemize}
Similarly, for the sets $V_{1}^{\SM}$, $V_{2+}^{\MD}$, and $V_{1}^{\MD}$, we have the same excess colors guarantee $\frac{\Delta}{2\log (1/\epsilon_i)}$ for each layer-$i$ vertex therein.

We apply the following lemmas to color the locally dense vertices $V^\star \setminus V_{\sparse}$;
refer to Section~\ref{sect:dense} for their proofs.
For small and medium blocks, we use Lemma~\ref{lem:color-sm} to color $V_{2+}^{\SM}$ and $V_{2+}^{\MD}$,
and use Lemma~\ref{lem:color-top-sm} to color $V_{1}^{\SM}$ and $V_{1}^{\MD}$.

The reason that the layer-1 blocks need to be treated differently is that layer-1 vertices do not obtain excess colors from the initial coloring step (Lemma~\ref{lem:initial-color}).
For comparison, for $i > 1$, each layer-$i$ vertex $v$ is $\epsilon_{i-1}$-sparse, and so $v$ must have $\Omega(\epsilon_{i-1}^2 \Delta) = \Omega(\epsilon_{i}^4 \Delta)$ excess colors. If we reduce the degree of $v$ to $\epsilon_i^5 \Delta$, then we obtain a sufficiently big gap between the excess colors and degree at $v$.

\begin{lemma}[Small and medium blocks; layers other than 1] \label{lem:color-sm}
Let $S = V_{2+}^{\SM}$ or $S = V_{2+}^{\MD}$.
Suppose that each layer-$i$ vertex $v \in S$ has at least $\frac{\Delta}{2\log (1/\epsilon_i)}$ excess colors w.r.t.~$S$.
There is an $O(1)$-time algorithm that colors a subset of $S$ meeting the following condition.
For each vertex  $v \in V^\star$, and for each $i \in [2, \ell]$,
with probability at least $1 - \exp(-\Omega(\poly(\Delta)))$,
the number of uncolored layer-$i$ neighbors of $v$ in $S$ is at most $\epsilon_i^5 \Delta$.
Vertices that violate this property join the set $V_{\bad}$.
\end{lemma}

\begin{lemma}[Small and medium blocks; layer 1] \label{lem:color-top-sm}
Let $S = V_1^{\SM}$ or $S = V_1^{\MD}$.
Suppose that each vertex $v \in S$ has at least $\frac{\Delta}{2\log (1/\epsilon_1)}$ excess colors w.r.t.~$S$.
There is an $O(1)$-time algorithm that colors a subset of $S$ meeting the following condition.
Each $v \in S$ is colored with probability at least $1 - \exp(-\Omega(\poly(\Delta)))$;
all uncolored vertices in $S$ join $V_{\bad}$.
\end{lemma}

The following lemmas consider large blocks.
Lemma~\ref{lem:color-lg} colors $V_{2+}^{\LG}$
and has guarantees similar to Lemma~\ref{lem:color-sm},
whereas Lemma~\ref{lem:color-top-lg} colors nearly all of $V_1^{\LG}$ and partitions the remaining uncolored vertices
among three sets, $R,X,$ and $V_{\bad}$, with certain guarantees.

\begin{lemma}[Large blocks; layer other than 1] \label{lem:color-lg}
There is an $O(1)$-time algorithm that colors a subset of $V_{2+}^{\LG}$ 
meeting the following condition.
For each $v \in V^\star$ and each layer number $i \in [2, \ell]$,
with probability at least $1 - \exp(-\Omega(\poly(\Delta)))$,
the number of uncolored layer-$i$ neighbors of $v$ in $V_{2+}^{\LG}$ 
is at most $\epsilon_i^5 \Delta$.
Vertices that violate this property join the set $V_{\bad}$.
\end{lemma}

Remember that our goal is to show that the bad vertices $V_{\bad}$ induce connected components of size $O(\poly \log n)$. However, if in a randomized procedure each vertex is added to $V_{\bad}$ with probability $1 - 1/\poly(\Delta)$, then the shattering lemma only guarantees that the size of each connected component of  $V_{\bad}$ is $O(\poly(\Delta)\log n)$, which is not necessarily $\poly\log n$.
This explains why Lemma~\ref{lem:color-top-lg} has two types of guarantees.

\begin{lemma}[Large blocks; layer 1] \label{lem:color-top-lg}
Let $c$ be any sufficiently large constant.
Then there is a constant time (independent of $c$) algorithm that colors a 
subset of $V_1^{\LG}$ while satisfying one of the following cases.  
\begin{itemize}
    \item The uncolored vertices of $V_1^{\LG}$ are partitioned among $R$ or $V_{\bad}$.
    The graph induced by $R$ has degree $O(c^2)$; each vertex joins $V_{\bad}$ with 
    probability $\Delta^{-\Omega(c)}$.
    \item If $\Delta > \log^{\alpha c} n$, where $\alpha > 0$ is some universal constant,
    then the uncolored vertices of $V_1^{\LG}$ are partitioned among $R$ and $X$,
    where the graph induced by $R$ has degree $O(c^2)$ and the components 
    induced by $X$ have size $\log^{O(c)} n$, w.h.p.
\end{itemize}
\end{lemma}


In our $(\Delta+1)$-list coloring algorithm, we apply Lemmas~\ref{lem:color-sm},~\ref{lem:color-top-sm},~\ref{lem:color-lg}, and~\ref{lem:color-top-lg} to color the vertices in $V^\star \setminus V_{\sparse}$, and they are processed in this order: $(V_{2+}^{\SM}, V_{1}^{\SM}, V_{2+}^{\MD}, V_{1}^{\MD}, V_{2+}^{\LG}, V_{1}^{\LG})$.

\paragraph{Coloring the Leftover Vertices $X$ and $R$.}
Notice that the algorithm for Lemma~\ref{lem:color-top-lg} generates a leftover uncolored subset $R$ which induces a constant-degree subgraph, and 
(in case $\Delta > \log^{\Theta(c)} n$)  
a leftover uncolored subset $X$ where
each connected component has size at most $O(\polylog n)$.
Remember that the vertices in $R$ and $X$ do not join $V_{\bad}$.
All vertices in $X$ are colored deterministically in 
$\Detd(\polylog n)$ time; the vertices in $R$ are 
colored deterministically in $O(\poly(\Delta') + \log^\ast n) = O(\log^\ast n)$ time~\cite{Linial92,FraigniaudHK16,BarenboimEG18}, with $\Delta' = O(c^2) = O(1)$.

\paragraph{The Remaining Vertices.}
Any vertex in $V^\star$ that violates at least one condition specified in the
lemmas is added to the set $V_{\bad}$.
All remaining uncolored vertices  join the set $U$.
In other words, $U$ is the set of all vertices in  
$V^\star \setminus (V_{\sparse} \cup V_{\bad} \cup R \cup X)$ that 
remain uncolored after applying the lemmas.

\subsection{Coloring the Remaining Vertices}

At this point all uncolored vertices are in $U\cup V_{\sparse} \cup V_{\bad}$.
We show that $U\cup V_{\sparse}$ can be colored efficiently in $O(\log^\ast \Delta)$ time  using  Lemma~\ref{lem:color-remain},
then consider $V_{\bad}$.

\paragraph{Coloring the Vertices in $U$.}
Let $G'$ be the directed acyclic graph
induced by $U$, where all edges are oriented from the sparser to the denser endpoint.
In particular, an edge $e=\{u,u'\}$ is oriented as $(u,u')$ if
$u$ is at layer $i$, $u'$ at layer $i'$, and $i > i'$,
or if $i=i'$ and $\ID(u)>\ID(u')$.
We write $\Nout(v)$ to denote the set of out-neighbors of $v$ in $G'$.

For each layer-$i$ vertex $v$ in $G'$, and each layer $j$,
the number of layer-$j$ neighbors of $v$ in $G'$ is at most $O(\epsilon_j^5 \Delta)$,
due to Lemmas~\ref{lem:color-sm} and~\ref{lem:color-lg}.
The out-degree of $v$ is therefore at most $\sum_{j=1}^{i} \epsilon_j^5 \Delta = O(\epsilon_i^5 \Delta) = O(\epsilon_{i-1}^{5/2} \Delta)$.

We write $\Psi(v)$ to denote the set of available colors of $v$.
The number of excess colors at $v$ is $|\Psi(v)| -  \deg(v) = \Omega(\epsilon_{i-1}^2 \Delta)$.
Thus, there is an $\Omega(1/\sqrt{\epsilon_{i-1}})$-factor gap between the palette size of $v$ and the out-degree of $v$.

Lemma~\ref{lem:color-remain} is applied to color nearly all vertices in $U$ in $O(\log^\ast \Delta)$ time,
with any remaining uncolored vertices added to $V_{\bad}$.
We use the following parameters of Lemma~\ref{lem:color-remain}.
In view of the above, there exists a constant $\eta > 0$ such that, for each $i \in [2,\ell]$
and each layer-$i$ vertex $v$ in $G'$, we set $p_v =  \eta \epsilon_{i-1}^2 \Delta \leq
|\Psi(v)| - \deg(v)$.
There is a constant $C > 0$ such that for each $i \in [2,\ell]$ and each layer-$i$ vertex $v \in U$, we have:
\[
\sum_{u \in  N_{\text{out}}(v)} 1 / p_u \,\leq\,  \sum_{j=2}^{i} O\paren{\frac{\epsilon_{j-1}^{5/2} \Delta}{\epsilon_{j-1}^2 \Delta}}  \,=\, \sum_{j=2}^{i} O({\epsilon_{j-1}^{1/2}}) \,<\,  1/C.
\]
The remaining parameters to Lemma~\ref{lem:color-remain} are
\begin{align*}
p^\star = \eta \epsilon_{1}^2 \Delta = \Omega(\Delta^{8/10}), \; \;
d^\star = \Delta, \; \;
C = \Omega(1).
\end{align*}
Thus, by Lemma~\ref{lem:color-remain} the probability that a vertex still remains uncolored (and is added to $V_{\bad}$) after the algorithm is
$$\exp(-\Omega(\sqrt{p^\star})) + d^\star \exp(-\Omega(p^\star)) =  \exp(-\Omega(\Delta^{2/5})).$$

\paragraph{Coloring the Vertices in $V_{\sparse}$.}  
The set $V_{\sparse}$ can be colored in a similar way using  Lemma~\ref{lem:color-remain}.
We let $G''$ be any acyclic orientation of the graph induced by $V_{\sparse}$, e.g., orienting each edge $\{u,v\}$ towards 
the vertex $v$ such that $\ID(u) > \ID(v)$.
The number of available colors of each $v \in V_{\sparse}$  minus its out-degree is at
least $\Omega(\epsilon_\ell^2 \Delta)$, which is at least $\gamma \Delta$, for some constant $\gamma > 0$, according to the way we select the sparsity sequence.
 We define $p_v = \gamma \Delta < |\Psi(v)| - \deg(v)$.
We  have $\sum_{u \in  N_{\text{out}}(v)} (1 / p_u) \leq \outdeg(v) / (\gamma \Delta) \leq 1/\gamma$.
Thus, we can apply Lemma~\ref{lem:color-remain} with $C = \gamma$.
Notice that both $p^\star$ and $d^\star$ are $\Theta(\Delta)$, and so the probability that a vertex still remains uncolored after the algorithm  (and is added to $V_{\bad}$)  is $\exp(-\Omega(\sqrt{\Delta}))$.

\paragraph{Coloring the Vertices in $V_{\bad}$.}
At this point, all remaining uncolored vertices are in $V_{\bad}$.
 If $\Delta \gg \polylog n$, then $V_{\bad} = \emptyset$, w.h.p., in view of the failure probabilities $\exp(-\Omega(\poly(\Delta)))$ specified in the lemmas used in the previous coloring steps. 
 Otherwise, $\Delta = O(\polylog n)$, and by Lemma~\ref{lem:shatter}, each connected component of
$V_{\bad}$ has size at most $\poly(\Delta)\log n  = O(\polylog n)$.
In any case, it takes $\Detd(\polylog n)$ to color all  vertices in $V_{\bad}$ deterministically.

\begin{figure}
\begin{center}
 \framebox{
\begin{minipage}{6in}
{\bf $(\Delta+1)$-List Coloring Algorithm}
\begin{enumerate}
    \item Determine the $\epsilon$-almost cliques, for $\epsilon\in\{\epsilon_1,\ldots,\epsilon_\ell\}$.  (Lemma~\ref{lem:cluster-property}.)
    
    \item Perform the initial coloring step using algorithm \oneshot{} (Lemma~\ref{lem:initial-color})
    and partition the remaining uncolored vertices into $V^\star$ and $V_{{\bad}}$.
    Further partition $V^\star$ into a sparse set $V_{\sparse}$ and a hierarchy $\mathcal{T}$ of small, medium, and large blocks.  Partition $V^\star \backslash V_{\sparse}$ into 6 sets: $V_{2+}^{\SM},V_1^{\SM},V_{2+}^{\MD},V_1^{\MD},
    V_{2+}^{\LG},V_1^{\LG}$.
 
    \item Color most of $V_{2+}^{\SM},V_1^{\SM},V_{2+}^{\MD},V_1^{\MD},
    V_{2+}^{\LG},V_1^{\LG}$ in six steps.
    \begin{enumerate}
     \item Color a subset of $V_{2+}^{\SM}$ using algorithm \dense{} (version 1). Any vertices that violate the conclusion of Lemma~\ref{lem:color-sm} are added to $V_{\bad}$.
     \item Color $V_1^{\SM}$ using algorithm \dense{} (version 1). Any remaining uncolored vertices are added to $V_{\bad}$ (Lemma~\ref{lem:color-top-sm}).
     \item Color a subset of $V_{2+}^{\MD}$ using algorithm \dense{} (version 1). Any vertices that violate the conclusion of Lemma~\ref{lem:color-sm} 
     are added to $V_{\bad}$.
     \item Color $V_1^{\MD}$ using algorithm \dense{} (version 1). Any remaining uncolored vertices are added to $V_{\bad}$ (Lemma~\ref{lem:color-top-sm}).
     \item Color a subset of $V_{2+}^{\LG}$ using algorithm \dense{} (version 2). 
     Any vertices that violate the conclusion of Lemma~\ref{lem:color-lg} are added to $V_{\bad}$.
     \item Color $V_1^{\LG}$ using algorithm \dense{} (version 2).  Each remaining uncolored vertex is added to one of $X,R,$ or $V_{\bad}$. (See Lemma~\ref{lem:color-top-lg}.)
    \end{enumerate}

    \item W.h.p.~$R$ induces a graph with constant maximum degree. 
    Color $R$ in $O(\log^\ast n)$ time deterministically using a standard algorithm~\cite{Linial92,FraigniaudHK16,BarenboimEG18}.
    
    \item W.h.p.~$X$ induces a graph whose components have size $\poly\log n$.
    Color $X$ in $O(\Detd(\poly\log n))$ time deterministically; see~\cite{PanconesiS96}.
    
    \item Color those uncolored vertices $U$ in 
    $\left(V_{2+}^{\SM}\cup V_{2+}^{\MD} \cup V_{2+}^{\LG}\right) \backslash V_{\bad}$
    in $O(\log^\ast \Delta)$ time using algorithm \trial{} (Lemma~\ref{lem:color-remain}).
    Any vertices in $U$ that remain uncolored are added to $V_{\bad}$.
    
    \item Color $V_{\sparse}$ in $O(\log^\ast \Delta)$ time using algorithm \trial{} (Lemma~\ref{lem:color-remain}). 
    Any vertices that remain uncolored are added to $V_{\bad}$.
    
    \item W.h.p.~$V_{\bad}$ induces components of size $\poly\log n$.
    Color $V_{\bad}$ in $O(\Detd(\poly\log n))$ time deterministically; see~\cite{PanconesiS96}.
\end{enumerate}
\end{minipage}
}
\end{center}
\caption{\label{fig:algorithm} Steps 1, 2, and 3(a--f) take constant time.
Steps 4, 6, and 7 take $O(\log^\ast n) = O(\Detd(\poly\log n))$ time~\cite{Linial92,Naor91}.  The bottleneck in the algorithm are
Steps 5 and 8, which take $O(\Detd(\poly\log n))$ time.  The algorithm
succeeds in the prescribed time, so long as the input to Steps 4, 5, and 8 are as they should be, i.e., inducing subgraphs with constant degree,
or $\poly\log n$-size components, respectively.  
(These are instances of $(\deg+1)$-list coloring.)
When $\Delta \gg \poly\log n$ is sufficiently large, the set
$V_{\bad}$ is empty, w.h.p., but $X$ may be non-empty, and induce
components with size $\poly\log n$.
}
\end{figure}

\medskip
See Figure~\ref{fig:algorithm} for a synopsis of every step
of the $(\Delta+1)$-list coloring algorithm. 

\subsection{Time Complexity}

The time for \oneshot{} (Fig.~\ref{fig:algorithm}, Step 2) is $O(1)$.
The time for processing each of $V_{2+}^{\SM}$, $V_{1}^{\SM}$, $V_{2+}^{\MD}$, $V_{1}^{\MD}$, $V_{2+}^{\LG}$, $V_{1}^{\LG}$ 
(Steps 3(a--f)) is $O(1)$.
Observe that each of Steps 2 and 3(a--f) may put vertices in $V_{\bad}$,
that Steps 3(a,c,e) leave some vertices uncolored, 
and that Step 3(f) also
puts vertices in special sets $X$ and $R$.
With high probability, $R$ induces components with constant degree,
which can be colored deterministically in $O(\log^\ast n)$ time (Step 4).
The uncolored vertices ($U$) from Steps 3(a,c,e) have a large gap 
between their palette size and degree, and can be 
colored in $O(\log^\ast \Delta)$ time using the \trial{} algorithm (Lemma~\ref{lem:color-remain}) in Step 6.  The same type of palette size-degree gap exists for $V_{\sparse}$ as well so \trial{} colors
it in $O(\log^\ast \Delta)$ time; for Step 7 we are applying Lemma~\ref{lem:color-remain} again, but with different parameters.

Finally, Steps 5 and 8 solve a $(\deg+1)$-list coloring problem on a graph whose components have size $\poly\log n$.  
Observe that $V_{\bad}$ is guaranteed to induce components with size $\poly(\Delta)\log n$, which happens to be $\poly\log n$ since
no vertices are added to $V_{\bad}$, w.h.p., 
if $\Delta \gg \poly\log n$ is sufficiently large.  
In contrast, in Step 5 $X$ can be non-empty even when $\Delta$ is large, but it still induces components with size $\poly\log n$.

Since $\log^\ast \Delta \le \log^\ast n = O(\Detd(\poly\log n))$~\cite{Linial92}, the bottleneck in the algorithm is solving $(\deg+1)$-list coloring, in Steps 5 and 8.

\begin{theorem}\label{thm:main}
In the $\LOCAL$ model, the $(\Delta+1)$-list coloring problem 
can be solved, w.h.p., in $O(\Detd(\polylog n))$ time, where
$\Detd(n')$ is the deterministic complexity 
of $(\deg+1)$-list coloring on $n'$-vertex graphs.
\end{theorem}

Next, we argue that if the palettes have $\poly\log n$ extra colors initially, we can list color the graph in $O(\log^\ast \Delta)$ time.

\begin{theorem}\label{thm:excess}
There is a universal constant $\gamma > 0$ such that 
the $(\Delta + \log^{\gamma} n)$-list coloring problem can be solved in the $\LOCAL$ model, w.h.p., in $O(\log^\ast \Delta)$ time.
\end{theorem}
\begin{proof}
For all parts of our $(\Delta+1)$-list coloring algorithm, 
except the first case of Lemma~\ref{lem:color-top-lg},
the probability that a vertex $v$ joins $V_{\bad}$ 
is $\exp(-\Omega(\poly (\Delta)))$.
Let $\alpha$ and $c$ be the constants in  Lemma~\ref{lem:color-top-lg}
and $k_1 = \Theta(c) \geq \alpha c$ be such that if $\Delta > \log^{k_1} n$, 
then the probability that a vertex $v$ joins $V_{\bad}$ in our $(\Delta+1)$-list coloring algorithm 
is $\exp(-\Omega(\poly (\Delta))) = 1/\poly(n)$.
Note that when  $\Delta > \log^{k_1} n$, 
no vertex is added to $V_{\bad}$ in  Lemma~\ref{lem:color-top-lg}.

Let $R' = R \cup X$ be the leftover vertices in Lemma~\ref{lem:color-top-lg} for the case $\Delta > \log^{k_1} n$.
There exists a constant $k_2 > 0$ such that the subgraph induced by
$R'$ has maximum degree $\log^{k_2} n$.
We set $\gamma = \max\{k_1, k_2\}+1$.
Now we show how to solve the $(\Delta + \log^{\gamma} n)$-list coloring problem in $O(\log^\ast \Delta)$ time.

If $\Delta \leq \log^{\gamma - 1} n$, then we apply the algorithm of Lemma~\ref{lem:color-remain-simple} directly, 
with $\rho = \frac{\log^{\gamma} n}{\Delta} - 1 = \Omega(\log n)$. The algorithm takes $O(1 + \log^\ast \Delta - \log^\ast \rho) = O(1)$ time,
and the probability that a vertex $v$ is not colored is $\exp(-\Omega(\sqrt{\rho \Delta})) = \exp(-\Omega(\log^{\gamma/2} n)) \ll 1/ \poly(n)$.

If $\Delta > \log^{\gamma - 1} n$, then we apply Steps 1,2,3,6, and 7 of our $(\Delta+1)$-list coloring algorithm.
Due to the lower bound on $\Delta$, we have $V_{\bad} = \emptyset$, w.h.p., which obviates the need to implement Step 8.

This algorithm takes $O(\log^\ast \Delta)$ time, and produces an 
uncolored subgraph $R' = R\cup X$ that has maximum degree 
$\Delta' \leq \log^{k_2} n$.  In lieu of Steps 4 and 5,
we apply the algorithm of Lemma~\ref{lem:color-remain-simple} to color $R'$ in $O(1 + \log^\ast \Delta' - \log^\ast \rho) = O(1)$ time, 
where $\rho = \frac{\log^\gamma n}{\Delta'} - 1 = \Omega(\log n)$.
\end{proof}

If every vertex is $\epsilon$-sparse,
with $\epsilon^2 \Delta$ sufficiently large, then the algorithm of Lemma~\ref{lem:initial-color}
gives \emph{every} vertex $\Omega(\epsilon^2 \Delta)$ excess colors, w.h.p.
Combining this observation with Theorem~\ref{thm:excess}, we have the following result, which shows that the
$(\Delta+1)$-list coloring problem can be solved very efficiently when all vertices are sufficiently locally sparse.

\begin{theorem}\label{thm:sparse}
There is a universal constant $\gamma > 0$ 
such that the following holds.
Suppose $G$ is a graph with maximum degree $\Delta$ in which each vertex is $\epsilon$-sparse, where $\epsilon^2\Delta > \log^\gamma n$.
A $(\Delta+1)$-list coloring of $G$ can be computed in the $\LOCAL$ model, w.h.p., in $O(\log^\ast \Delta)$ time.
\end{theorem}

\begin{remark}
Theorem~\ref{thm:sparse} insists on every vertex being
\emph{$\epsilon$-sparse} according to Definition~\ref{def-sparse-2}.
It is straightforward to show connections between this definition of sparsity and
others standard measures from the
literature.  For example, such a graph is
\emph{$(1-\epsilon')$-locally sparse},
where $\epsilon'=\Omega(\epsilon^2)$, according to Definition~\ref{def-sparse-1}.
Similarly, any $(1-\epsilon')$-locally sparse graph is $\Omega(\epsilon')$-sparse.
Graphs of \emph{degeneracy} $d\le (1-\epsilon')\Delta$
or \emph{arboricity} $\lambda \le (1/2 - \epsilon')\Delta$ are trivially
$(1-\Omega(\epsilon'))$-locally sparse. 
\end{remark}

\begin{remark}
We have made no effort to minimize the constant $\gamma$ in Theorems~\ref{thm:excess} and \ref{thm:sparse}, 
and it is impractically large.  It would be useful to know
if these theorems remain true when $\gamma$ is small, say 1,
i.e., is $(\Delta+\log n)$-coloring solvable in 
$O(\log^\ast \Delta)$ time, w.h.p.?
\end{remark}

\section{Fast Coloring using Excess Colors}\label{sect:trial-detail}

In this section, we prove Lemma~\ref{lem:color-remain}.
Consider a directed acyclic graph $G=(V,E)$, where each vertex $v$ has a palette $\Psi(v)$.
Each vertex $v$ is associated with a parameter $p_v \leq |\Psi(v)| - \deg(v)$, i.e., $p_v$ is a lower bound on the number of excess colors at $v$.
All vertices agree on values 
$p^\star \le \min_{v\in V} p_v$,
$d^\star \ge \max_{v \in V} \outdeg(v)$,
and $C = \Omega(1)$, such that the following is satisfied for all $v$.
\begin{equation}\label{eqn:Cpu}
\sum_{u \in  N_{\text{out}}(v)}  1/ p_u \leq 1/C.
\end{equation}
Intuitively, the sum $\sum_{u \in  N_{\text{out}}(v)}  1/ p_u$ measures the amount of ``contention'' at a vertex $v$.
In the \trial{} algorithm each vertex $v$ 
selects each color $c \in \Psi(v)$ with probability $\frac{C}{2  |\Psi(v)|} < \frac{C}{2  p_v}$ and permanently colors itself if it selects a color
not selected by any out-neighbor.
\begin{framed}
\noindent {\bf Procedure} \trial.
\begin{enumerate}
\item Each color $c \in \Psi(v)$ is added to $S_v$ independently with probability $\frac{C}{2  |\Psi(v)|}$.
\item If there exists a color 
$\displaystyle c^\star \in S_v\backslash \left(\bigcup_{u\in \Nout(v)} S_u\right)$, $v$ permanently colors itself $c^\star$.
\end{enumerate}
\end{framed}
In Lemma~\ref{lem:shrink} we present an analysis of \trial.
We show that after an iteration of \trial, the amount of ``contention'' at a vertex $v$ decreases by (roughly) an  $\exp(C/6)$-factor, with very high probability.

\begin{lemma}\label{lem:shrink}
Consider an execution of \trial.
Let $v$ be any vertex. Let $D$ be the summation of $1/p_u$ over all vertices $u$ in $\Nout(v)$ that remain uncolored after \trial.
Then the following holds.
\begin{align*}
\Prob[\text{ $v$ remains uncolored }] &\leq \exp(-C/6) + \exp(-\Omega(p^\star)). \\
\Prob[D \geq (1+\lambda)  \exp(-C/6) / C ] &\leq \exp\left(-2 \lambda^2   p^\star \exp(-C/3) / C \right) +d^\star\exp(-\Omega(p^\star)).
\end{align*}
\end{lemma}
\begin{proof}
For each vertex $v$, we define the following two events.
\begin{description}
\item $E_v^{\text{good}} \, :$  $v$ selects a color that is not selected by any vertex in $N_{\text{out}}(v)$.
\item $E_v^{\text{bad}}\, :$  the number of colors in $\Psi(v)$ that are selected by some vertices in $N_{\text{out}}(v)$ is at least $\frac23 \cdot |\Psi(v)|$.
\end{description}
Notice that $E_v^{\text{good}}$ is the event where $v$ successfully colors itself.
We first show that $\Prob[E_v^{\text{bad}}] =  \exp(-\Omega(p^\star))$. Fix any color $c \in \Psi(v)$.
The probability that $c$ is selected by some vertex in $N_{\text{out}}(v)$ is
\[
1 - \prod_{u\in N_{\text{out}}(v)} \paren{1 - \fr{C}{2 |\Psi(u)|}}
\le
1 - \prod_{u\in N_{\text{out}}(v)} \paren{1 - \fr{C}{2p_u}} \le \sum_{u\in N_{\text{out}}(v)} \fr{C}{2p_u} \le \fr{1}{2},
\]
where the last inequality follows from (\ref{eqn:Cpu}).
Since these events are independent for different colors, 
$\Prob[E_v^{\text{bad}}] \leq \Prob[\text{Binomial}(n', p') \geq \frac{2 n'}{3}]$ with $n' = |\Psi(v)| \geq p_{v}$ and $p' = \frac12$. 
By a Chernoff bound, we have:
\[
\Prob\left[E_u^{\text{bad}}\right] \leq \exp(-\Omega(n' p')) =  \exp(-\Omega(p^\star)).
\]
Conditioned on $\overline{E_v^{\text{bad}}}$, $v$ will color itself unless it fails to choose \emph{any} of
$|\Psi(v)|/3$ specific colors from its palette.
Thus,
\begin{equation}\label{eqn:Egood}
\Prob\left[\left.\overline{E_v^{\text{good}}} \;\right|\; \overline{E_v^{\text{bad}}}\right]
\leq \paren{1 - \fr{C}{2|\Psi(v)|}}^{|\Psi(v)|/3} \leq \exp\paren{{\fr{-C}{6}}}.
\end{equation}

\medskip

We are now in a position to prove the first inequality of the lemma.
The probability that $v$ remains uncolored is 
at most $\Prob\left[\left.\overline{E_v^{\text{good}}} \;\right|\; \overline{E_v^{\text{bad}}}\right]
+ \Prob\left[E_v^{\text{bad}}\right]$, which is at most $\exp(-C/6) + \exp(-\Omega(p^\star))$.

\medskip

Next, we prove the second inequality, on the upper
tail of the random variable $D$.
Let $N_{\text{out}}(v) = (u_1, \ldots, u_k)$.
Let $E_i^{\text{bad}}$ and $E_i^{\text{good}}$ be short for $E_{u_i}^{\text{bad}}$ and $E_{u_i}^{\text{good}}$,
and let $\mathcal{E}$ be the event $\bigcup_i E_i^{\text{bad}}$.
By a union bound,
\[
\Prob\sqbrack{\mathcal{E}} \leq \outdeg(v) \cdot \exp(-\Omega(p^\star)) \leq d^\star \cdot \exp(-\Omega(p^\star)).
\]

Let ${X} = \sum_{i=1}^k X_i$, where 
$X_i = 1/{p_{u_i}}$ if  \emph{either}
$\overline{E_i^{\text{good}}}$ or $E_i^{\text{bad}}$ occurs, 
and $X_i = 0$ otherwise.
Observe that if we condition on 
$\overline{\mathcal{E}}$,
then ${X}$ is exactly $D$, 
the random variable we want to bound.

By linearity of expectation,
\begin{align*}
    \mu = \Expect[X \;|\; \overline{\mathcal{E}}] 
    &= \sum_i \Expect[X_i \;|\; \overline{\mathcal{E}}]\\
    &\leq \sum_i \frac{1}{p_{u_i}}\cdot \Prob\left[\left.\overline{E_i^{\text{good}}} \;\right|\; \overline{E_i^{\text{bad}}}\right]\\
    &\leq \sum_i \exp(-C/6)/p_{u_i} & \mbox{Equation~(\ref{eqn:Egood})}\\
    &\leq \exp(-C/6)/C     & \mbox{Equation~(\ref{eqn:Cpu})}
\end{align*}

Each variable $X_i$ is within the range $[a_i, b_i]$, where $a_i = 0$ and $b_i = 1/{p_{u_i}}$.
We have $\sum_{i=1}^k (b_i - a_i)^2 \leq  \sum_{u \in  N_{\text{out}}(v)} 1/ (p_u \cdot p^\star) \leq 1/ (C p^\star)$.
By Hoeffding's inequality,\footnote{The variables $\{X_1, \ldots, X_k\}$ are not independent, but we are still able to apply Hoeffding's inequality. The reason is as follows. Assume that $N_{\text{out}}(v) = (u_1, \ldots, u_k)$ is sorted in reverse topological order, and so for each $1 \leq j \leq k$, we have $N_{\text{out}}(u_j)\cap \{u_{j}, \ldots, u_k\} = \emptyset$. Thus, conditioning on (i) $\overline{E_i^{\text{bad}}}$ and (ii) {\em any} colors selected by vertices in $\bigcup_{1 \leq j < i} N_{\text{out}}(u_j) \cup \{u_j\}$, the probability that $\overline{E_i^{\text{good}}}$ occurs is still at most $\exp({\frac{-C}{6}})$.}
we have
\begin{align*}
\Prob[X \geq (1+\lambda) \exp(-C/6) / C \;|\; \overline{\mathcal{E}}]
&\leq \Prob[X  \geq (1+\lambda)\mu \;|\; \overline{\mathcal{E}}] \\
&\leq \exp\left(\frac{-2(\lambda \mu)^2}{\sum_{i=1}^k (b_i - a_i)^2}\right) \\
& \leq \exp\left(-2 (\lambda \exp(-C/6) / C)^2 (p^\star C)\right) \\
& = \exp\left(-2 \lambda^2  p^\star \exp(-C/3) / C\right).
\end{align*}

Thus,
\begin{align*}
\Prob[D \geq (1+\lambda) \exp(-C/6) / C ]
&\leq
\Prob[X \geq (1+\lambda) \exp(-C/6) / C \;|\; \overline{\mathcal{E}}] + 
\Prob[\mathcal{E}]\\
&\leq 
\exp\left(-2 \lambda^2  p^\star \exp(-C/3) / C\right) + 
d^\star\exp(-\Omega(p^\star)). 
\qedhere
\end{align*}
\end{proof}

\begin{proof}[Proof of Lemma~\ref{lem:color-remain}]
In what follows, we show how Lemma~\ref{lem:shrink} can be used to derive  Lemma~\ref{lem:color-remain}.
Our plan is to apply \trial\ for 
$k^\star = \log^\ast p^\star - \log^\ast C + O(1)$ iterations.
For the $k$th iteration we use the parameter $C_k$, 
which is defined as follows:
\begin{align*}
    C_1 &= \min\{\sqrt{p^\star},\; C\},\\
    C_k &= \min\left\{\sqrt{p^\star},\; \frac{C_{k-1}}{(1+ \lambda)\exp(-C_{k-1}/6)}\right\}\\
    k^\star &= \min\{k \;|\; C_k = \sqrt{p^\star}\} & \mbox{(the last iteration)}
\end{align*}
Here $\lambda > 0$ must be selected to be a sufficiently small constant so that $(1+\lambda)\exp(-C_{k-1}/6) < 1$. 
This guarantees that the sequence $(C_k)$ is strictly increasing.
For example, if $C \ge 6$ initially, 
we can fix $\lambda=1$ throughout.

We analyze each iteration of  \trial{} 
using the same (initial) vector of $(p_v)$ values, i.e., we do not
count on the number of excess colors at any vertex increasing over time.

At the \emph{end} of the $k$th iteration, $k\in [1,k^\star]$, we have the following invariant $\mathcal{H}_k$ that we expect all vertices to satisfy:
\begin{itemize}
\item If $k\in [1,k^\star)$, $\mathcal{H}_k$ stipulates that
for each uncolored vertex $v$ after the $k$th iteration, 
the summation of $1/p_u$ over all uncolored $u\in \Nout(v)$ is 
less than $1 / C_{k+1}$.
\item $\mathcal{H}_{k^\star}$ stipulates that all vertices are colored at the end of the $k^\star$th iteration.
\end{itemize}
The purpose of $\mathcal{H}_k$, $k\in [1,k^\star)$, 
is to guarantee that $C_{k+1}$ is a valid parameter for the 
$(k+1)$th iteration of \trial.
For each $k\in [1,k^\star]$, at the end of the $k$th iteration we remove from consideration all vertices violating $\mathcal{H}_k$,
and add them to the set $V_{\bad}$. Thus, by definition of $\mathcal{H}_{k^\star}$, after the last iteration, all vertices other than the ones in $V_{\bad}$ have been colored.

To prove the lemma, it suffices to show that 
the probability of $v$ joining $V_{\bad}$ is 
at most $\exp(-\Omega(\sqrt{p^\star})) + d^\star \exp(-\Omega(p^\star))$,
and this is true even if the randomness outside a 
constant radius around $v$ is determined adversarially.
By Lemma~\ref{lem:shrink}, the probability that a vertex is removed at the end of the $k$th iteration, 
where $k\in [1,k^\star)$,
is at most
\begin{align*}
& \exp(\Omega(p^\star / C_{k+1})) + d^\star \exp(-\Omega(p^\star))
\leq \exp(-\Omega(\sqrt{p^\star})) + d^\star \exp(-\Omega(p^\star)).
\end{align*}
The probability that a vertex is removed at the end of the $k^\star$th iteration is at most
$\exp(-C_{k^\star}/6) + \exp(-\Omega(p^\star)) \leq \exp(-\Omega(\sqrt{p^\star}))$.
By a union bound over all 
$k^\star = \log^\ast p^\star - \log^\ast C + O(1)$ iterations, 
the probability that a vertex joins $V_{\bad}$ 
is $\exp(-\Omega(\sqrt{p^\star})) + d^\star \exp(-\Omega(p^\star))$.
\end{proof}

\section{Coloring Locally Dense Vertices}\label{sect:dense}

Throughout this section, we consider the following setting.
We are given a graph $G=(V,E)$, where some vertices are already colored.  We are also given a
\emph{subset}  $S$  of the uncolored vertices, which
is partitioned into $g$ disjoint {\em clusters} $S = S_1 \cup S_2 \cup \cdots \cup S_g$, each with weak diameter 2.
(In particular, this implies that otherwise sequential algorithms
can be executed on each cluster in $O(1)$ rounds in the $\LOCAL$ model.)
Our goal is to color a large fraction of the vertices in $S$ 
in only constant time.

We assume that the edges within $S$ are oriented from the sparser to the denser endpoint, breaking ties by comparing IDs.
In particular, an edge $e=\{u,u'\}$ is oriented as $(u,u')$ if
$u$ is at layer $i$, $u'$ at layer $i'$, and $i > i'$,
or if $i=i'$ and $\ID(u)>\ID(u')$.
Notice that this orientation is {\em acyclic}.
We write $\Nout(v)\subseteq S$ to denote the set of 
out-neighbors of $v$ in $S$.

In Section~\ref{sect:manyexcesscolors} we describe a procedure \dense{} (version 1)
that is efficient when each vertex has many excess colors w.r.t.~$S$.  It is analyzed
in Lemma~\ref{lem:dense-100}, which is then used to prove Lemmas~\ref{lem:color-sm} and \ref{lem:color-top-sm}.
In Section~\ref{sect:noexcesscolors} we describe a procedure \dense{} (version 2),
which is a generalization of 
Harris, Schneider, and Su's~\cite{HarrisSS18} procedure.
It is analyzed in Lemma~\ref{lem:dense-101}, which is then used 
to prove Lemmas~\ref{lem:color-lg} and \ref{lem:color-top-lg}.

\subsection{Version 1 of \dense\ --- Many Excess Colors are Available}\label{sect:manyexcesscolors}

In this section we focus on the case where  each vertex $v \in S$ has many excess colors w.r.t.~$S$.
We make the following assumptions about the vertex set $S$.
\begin{description}
\item[Excess colors.] Each $v \in S$ is associated with a parameter $Z_v$, which indicates a lower bound on the number of excess colors of $v$ w.r.t.~$S$. That is, the palette size of $v$ minus $|N(v) \cap S|$ is at least $Z_{v}$.
\item[External degree.]  For each cluster $S_j$, each vertex $v \in S_j$ is associated with a parameter $D_v$ such that $|\Nout(v) \cap (S \setminus S_j)| \leq D_v$. 
\end{description}
The ratio of these two quantities plays an important role in the analysis.  Define $\delta_v$ as 
\[
\delta_v = D_v / Z_v.
\]
We briefly explain how we choose the clustering  $S = S_1 \cup S_2 \cup \cdots \cup S_g$ and set these parameters in the settings of Lemma~\ref{lem:color-sm} and Lemma~\ref{lem:color-top-sm}.
For Lemma~\ref{lem:color-top-sm}, $S$ is either $V_1^{\SM}$ or $V_1^{\MD}$, and each cluster of $S$ is the intersection of $S$ and an $\epsilon_1$-almost clique (a layer-1 block).
For Lemma~\ref{lem:color-sm}, $S$ is either $V_{2+}^{\SM}$ or $V_{2+}^{\MD}$, and each cluster of $S$ 
is the intersection of $S$ and an $\epsilon_{\ell}$-almost clique.
In all cases, clusters have weak diameter 2.
All vertices in the same layer adopt the same $D$- and $Z$-values. 
A layer-$i$ vertex $v$ takes 
\begin{align*}
Z_{v} &= \frac{\Delta}{2\log (1/\epsilon_i)}, \\
\mbox{ and } D_v &= \epsilon_i\Delta.
\end{align*}
The choices of these parameters are valid in view of the excess colors implied by Lemma~\ref{lem:palette-lb} 
and the external degree upper bound of Lemma~\ref{lem:cluster-property}.

\begin{framed}
\noindent {\bf Procedure} \dense~(version 1).
\begin{enumerate}
\item  Let $\pi : \{1,\ldots,|S_j|\} \rightarrow S_j$ be the \underline{unique} permutation that lists
$S_j$ in increasing order by layer number, breaking ties
(within the same layer) by ID.
For $q$ from $1$ to $|S_j|$, the vertex $\pi(q)$ selects a color $c(\pi(q))$ uniformly at random from
\[
\Psi(\pi(q)) \setminus \left\{c(\pi(q')) \ |\ q'<q \mbox{ and } \{\pi(q),\pi(q')\}\in E(G)\right\}.
\]
 \item Each $v \in S_j$ permanently colors itself $c(v)$ if $c(v)$ is not selected by any vertices in $\Nout(v)$.
\end{enumerate}
\end{framed}

Notice that $\pi$ is a reverse topological ordering of $S_j$, 
i.e., if $\pi(q')$ precedes $\pi(q)$, then 
$\pi(q) \notin \Nout(\pi(q'))$.
Because each $S_j$ has weak diameter 2, Step 1 of $\dense$ can be simulated with only $O(1)$ rounds of communication.
Intuitively, the probability that a vertex $v \in S$ remains uncolored after \dense~(version 1) is at most $\delta_v$,
since it is \emph{guaranteed} not to have any conflicts with neighbors in the same cluster. 
The following lemma gives us the probabilistic guarantee of the \dense~(version 1).

\begin{lemma}\label{lem:dense-100}
Consider an execution of \dense~(version 1). Let $T$ be any subset of $S$, and let $\delta = \max_{v \in T} \delta_v$.
For any $t\ge 1$, the number of uncolored vertices in $T$ is at least $t$ 
with probability at most $\Prob[\operatorname{Binomial}(|T|, \delta) \geq t]$.
\end{lemma}

\begin{proof}
Let $T=\{v_1, \ldots, v_{|T|}\}$ be listed in increasing order by layer number, breaking ties by vertex ID.
Remember that vertices in $T$ can be spread across multiple clusters in $S$.
Imagine exposing the color choices of all vertices in $S$, one by one, in this order $v_1, \ldots, v_{|T|}$.
The vertex $v_k$ in cluster $S_j$ will successfully color itself if it chooses any color not already selected by a vertex in
$\Nout(v_k) \cap (S\setminus S_j)$.  Since $|\Nout(v_k)\cap (S \setminus S_j)| \le D_{v_k}$ and $v_k$ has at least $Z_{v_k}$ colors to choose from at this moment, 
the probability that it fails to be colored is at most $D_{v_k}/Z_{v_k} = \delta_{v_k} \le \delta$,
\emph{independent of the choices made by higher priority vertices $v_1, \ldots, v_{k-1}$}.
Thus, for any $t$, the number of uncolored vertices in $T$ is stochastically dominated by 
the binomial variable $\operatorname{Binomial}(|T|, \delta)$.
\end{proof}

\begin{proof}[Proof of Lemma~\ref{lem:color-sm}]
We execute \dense~(version 1) for 6 iterations, where each
participating vertex $x\in S$ 
uses the same (initial) values of $Z_x$ and $D_x$, 
namely $Z_{x} = \frac{\Delta}{2\log (1/\epsilon_i)}$ and $D_x = \epsilon_i\Delta$ if $x$ is at layer $i$.

Consider any vertex $v \in V^\star$, 
and any layer number $i \in [2,\ell]$.
Let $T$ be the set of layer-$i$ neighbors of $v$ in $S$.
To prove Lemma~\ref{lem:color-sm}, it suffices to show that after 6 iterations of \dense~(version 1), with probability $1 - \exp(-\Omega(\poly(\Delta)))$, the number of uncolored vertices in $T$ is  at most $\epsilon_i^5 \Delta$.

We define the following parameters.
\begin{align*}
\delta              &= \max_{u \in T}\{\delta_u\} \:=\:  2\epsilon_i\log(1/\epsilon_i),\\
t_1                 &= |T|,\\
\mbox{ and } t_k    &= \max\left\{(2\delta)t_{k-1}, \epsilon_i^5 \Delta\right\}.
\end{align*}
Since $(2\delta)^6 |T| \leq \epsilon_i^5 \Delta$, we have $t_7 = \epsilon_i^5 \Delta$.

Assume that at the beginning of the $k$th iteration, 
the number of uncolored vertices in $T$ is at most $t_{k}$.
Indeed for $k = 1$, we initially have $t_1 = |T|$.
By Lemma~\ref{lem:dense-100}, after the $k$th iteration,
the expected number of uncolored vertices in $T$ is at most $\delta t_{k} \leq t_{k+1}/2$.
By a Chernoff bound,
with probability at most
$\exp(-\Omega(t_{k+1})) \leq \exp(-\Omega(\epsilon_i^5 \Delta)) = \exp(-\Omega(\poly(\Delta)))$, the number of uncolored vertices in $T$ is more than $t_{k+1}$.

Therefore, after 6 iterations of \dense~(version 1), with probability $1 - \exp(-\Omega(\poly(\Delta)))$, the number of uncolored vertices in $T$ is  at most $t_7 = \epsilon_i^5 \Delta$, as required.
\end{proof}

\begin{proof}[Proof of Lemma~\ref{lem:color-top-sm}]
In the setting of Lemma~\ref{lem:color-top-sm} we only consider layer-1 vertices, but have the higher burden of coloring \emph{each} vertex with high enough probability. 
Since $\epsilon_1 = \Delta^{-1/10}$, we have
$Z_{v} = \frac{\Delta}{2\log (1/\epsilon_1)}$,  $D_v = \epsilon_1\Delta$, and $\delta_v = D_v / Z_{v} = 2\epsilon_1 \log (1 / \epsilon_1)$, 
for all vertices $v \in S$.

We begin with one iteration of \dense~(version 1).
By Lemma~\ref{lem:dense-100} and a Chernoff bound, for each $v \in S$, the number of uncolored vertices of $N(v) \cap S$ is at most  
$2 \delta_v \Delta = \Delta' < O(\Delta^{9/10}\log\Delta)$ with probability $1 - \exp(-\Omega(\poly(\Delta)))$.
Any uncolored vertex $v \in S$ that violates this property,
i.e., for which $|N(v)\cap S| > \Delta'$, is
added to $V_{\bad}$ and removed from further consideration.

Consider the graph $G'$ induced by the  remaining uncolored vertices in $S$.
The maximum degree of $G'$ is at most $\Delta'$.
Each vertex $v$ in $G'$ satisfies  $|\Psi(v)| \geq Z_{v} = \frac{\Delta}{2\log (1/\epsilon_1)} = (1+\rho)\Delta'$, 
where $\rho$ is $\Delta^{\Omega(1)}$.
We run the algorithm of Lemma~\ref{lem:color-remain-simple} on 
$G'$, and then put all vertices that still remain uncolored to 
the set  $V_{\bad}$.
By Lemma~\ref{lem:color-remain-simple}, the time for this procedure is $O(\log^\ast \Delta - \log^\ast \rho) = O(1)$, and the probability that a vertex $v$ remains uncolored and is added to $V_{\bad}$ is at most $\exp(-\Omega(\sqrt{\rho\Delta})) = \exp(-\Omega(\poly(\Delta)))$.
\end{proof}

\subsection{Version 2 of \dense\ --- No Excess Colors are Available}\label{sect:noexcesscolors}

In this section we focus on the case where there is no guarantee on the 
number of excess colors.
The palette size lower bound of each vertex $v \in S_j$ comes from the assumption
that $|S_j|$ is large, and $v$ is adjacent to all but a very small portion of vertices in $S_j$.
For the case $S = V_{2+}^{\LG}$ (Lemma~\ref{lem:color-lg}), 
each cluster $S_j$ is a large block in some layer $i\in [2,\ell]$.
For the case $S = V_{1}^{\LG}$ (Lemma~\ref{lem:color-top-lg}), 
each $S_j$ is a layer-$1$ large block.
For each $v \in S$, we define $\Nstar(v)$ to be the set of all vertices $u \in N(v) \cap S$ such that the layer number of $u$ 
is smaller than or equal to the layer number of $v$. 
Observe that $\Nout(v) \subseteq \Nstar(v)$ since 
$\Nout(v)$
excludes some vertices at $v$'s layer, 
depending on the ordering of IDs.
For the case of $S = V_{1}^{\LG}$, all clusters  $S_1,\ldots,S_g$  are layer-1 blocks, and so $\Nstar(v) = N(v) \cap S$.
We make the following assumptions.
\begin{description}
  \item[Identifiers.] List the clusters $S_1,\ldots,S_g$ 
  in non-decreasing order by layer number.
  We assume each cluster and each vertex within a cluster has an ID
  that is consistent with this order, in particular:
  \begin{align*}
      &\ID(S_1) < \cdots < \ID(S_g)\\
      &\max_{v\in S_j} \ID(v) < \min_{u\in S_{j+1}} \ID(u), \mbox{ for all $j\in [1,g)$}
  \end{align*}
   Given arbitrary IDs, it is straightforward to compute new IDs satisfying these properties in $O(1)$ time.
   (It is not required that each cluster $S_j$ to know the index $j$.)
  \item[Degree upper bounds.] Each cluster $S_j$ is associated with a parameter $D_j$ such that all $v \in S_j$ satisfy the following two conditions:
  \begin{enumerate}
      \item[(i)] $|S_j \setminus (N(v) \cup \{v\})| = |S_j \setminus (\Nstar(v) \cup \{v\})| \leq D_j$ (anti-degree upper bound),
      \item[(ii)] $|\Nstar(v) \setminus S_j| \leq D_j$ (external degree upper bound).
  \end{enumerate}
  \item[Shrinking rate.]  Each cluster $S_j$ is associated with a parameter  $\delta_j$  such that
  \[
   1/K \geq \delta_j \geq \frac{D_j \log(|S_j|/D_j)}{|S_j|},
  \]
  for some sufficiently large constant $K$.
\end{description}
The procedure
\dense{}~(version 2) aims to successfully color
a large  fraction of the vertices in each cluster $S_j$.
In Step 1, each cluster selects a $(1-\delta_j)$-fraction of its vertices
uniformly at random, permutes them randomly, and marches through
this permutation one vertex at a time.  As in \dense{} (version 1), 
when a vertex $v$ is processed it picks a random color $c(v)$ 
from its available palette that were 
not selected by previously processed vertices in $S_j$.  
Step 2 is the same: if $c(v)$ has not been selected by any neighbors of $\Nout(v)$ it permanently commits to $c(v)$.  There are only two reasons a vertex in $S_j$ may be left uncolored by \dense{} (version 2): 
it is not among the $(1-\delta_j)$-fraction of vertices
participating in Step 1, or it has a color conflict with an external neighbor in Step 2.  The first cause occurs with probability $\delta_j$
and, intuitively, the second cause occurs with probability about $\delta_j$ because vertices typically have \emph{many} options for colors when they are processed but \emph{few} external neighbors that can generate conflicts.  
Lemma~\ref{lem:dense-101} captures this formally; it is the culmination and corollary of Lemmas~\ref{lem:color-prob}--\ref{lem:uncolor1}, which
are proved later in this section.
Lemma~\ref{lem:dense-101} is used to prove Lemmas~\ref{lem:color-lg} and \ref{lem:color-top-lg}.

\begin{framed}
\noindent {\bf Procedure} \dense~(version 2).
\begin{enumerate}
\setlength{\itemsep}{-2pt}

\item  Each cluster $S_{j}$ selects $(1 - \delta_{j})|S_j|$ vertices u.a.r.~and generates
a permutation $\pi$ of those vertices u.a.r.
The vertex $\pi(q)$ selects a color $c(\pi(q))$ u.a.r.~from
\[
\Psi(\pi(q)) -  \left\{c(\pi(q')) \ |\ q'<q \mbox{ and } \{\pi(q),\pi(q')\}\in E(G)\right\}.
\]
 \item Each $v \in S_j$ that has selected a color $c(v)$
 permanently colors itself $c(v)$ if $c(v)$ is not selected by any vertices $u \in \Nout(v)$.
\end{enumerate}
\end{framed}

\medskip 

\begin{lemma}\label{lem:dense-101}
Consider an execution of \dense~(version 2).
Let $T$ be any subset of $S$, and let $\delta = \max_{j: S_j \cap T \neq \emptyset} \delta_j$.
For any number $t$,
the probability that the number of uncolored vertices in $T$ is at least $t$ is at most ${|T| \choose t} \cdot \left( O(\delta)\right)^{t}$.
\end{lemma}

Our assumption about the identifiers of clusters and vertices guarantees that for each $v \in S_j$, we have $\Nout(v) \subseteq \bigcup_{i=1}^{j} S_i$. Therefore, in the proof of Lemma~\ref{lem:dense-101}, we  expose the random bits of the clusters in the order $(S_1, \ldots, S_g)$. 
Once the random bits of $S_1, \ldots, S_j$ are revealed, we can determine whether any particular $v \in S_j$ successfully colors itself.

Our proofs of Lemmas~\ref{lem:color-lg} and \ref{lem:color-top-lg} are based on a constant number of iterations of \dense~(version 2).
In each iteration, 
the parameters $D_j$ and $\delta_j$ might be different.
In subsequent discussion, the term {\em anti-degree} of $v \in S_j$ refers to the number of uncolored vertices in $S_j \setminus (N(v) \cup \{v\})$, and the term {\em external degree} of $v \in S_j$ refers to the number of uncolored vertices in  $\Nstar(v) \setminus S_j$.
Suppose $S_j$ is a layer-$i$ large block. The parameters for $S_j$ in each iteration are as follows.
Let $\beta > 0$ be a sufficiently large constant to be determined.
\begin{description}
 \item[Degree upper bounds.] By Lemma~\ref{lem:cluster-property}, 
 $D_j^{(1)} = 3 \epsilon_i \Delta$ upper bounds the initial anti-degree and external degree. For $k>1$, the parameter $D_j^{(k)}$ is chosen such that  $D_j^{(k)} \geq \beta \delta_j^{(k-1)}  \cdot D_j^{(k-1)}$. We write $\mathcal{D}_j^{(k)}$ to denote the invariant that at the \emph{beginning} of the $k$th iteration, $D_j^{(k)}$ is an upper bound on the anti-degree and external degree of all uncolored vertices in $S_j\setminus V_{\bad}$.
 \item[Cluster size upper bounds.] By Lemma~\ref{lem:cluster-property},
 $U_j^{(1)} = (1 + 3 \epsilon_i)\Delta$ is an upper bound on the initial cluster size. For $k>1$, the parameter  $U_j^{(k)}$ is chosen such that $U_j^{(k)} \leq  \beta \delta_j^{(k-1)}  \cdot U_j^{(k-1)}$. We write $\mathcal{U}_j^{(k)}$ to denote the invariant that at the beginning of the $k$th iteration, the number of uncolored vertices in $S_j\setminus V_{\bad}$ is at most $U_j^{(k)}$.
 \item[Cluster size lower bounds.] $L_j^{(1)} = \frac{\Delta}{\log (1/\epsilon_i)}$. For $k>1$, the parameter $L_j^{(k)}$  is chosen such that  $L_j^{(k)} \geq  \delta_j^{(k-1)}  \cdot L_j^{(k-1)}$.  We write $\mathcal{L}_j^{(k)}$ to denote the invariant that at the beginning of the $k$th iteration, the number of uncolored vertices in $S_j \setminus V_{\bad}$ is at least $L_j^{(k)}$.  By the definition of \emph{large} blocks,
 $\mathcal{L}_j^{(1)}$ holds initially.
 \item[Shrinking rates.] For each $k$, the shrinking rate  $\delta_j^{(k)}$ of  cluster $S_j$ for the $k$th iteration is chosen such that
 \[
 1/K \geq \delta_j^{(k)} \geq \frac{D_j^{(k)} \log{(L_j^{(k)} / D_j^{(k)})}}{L_j^{(k)}}.
 \]
 Additionally, we require that 
 $\delta_1^{(k)} \leq \cdots \leq \delta_g^{(k)}$,
 with 
 $\delta_j^{(k)} = \delta_{j+1}^{(k)}$ if $S_j$ and $S_{j+1}$ are in the same layer.
\end{description}

Although the initial values of $D_j^{(1)}, U_j^{(1)}, L_j^{(1)}$ are determined, there is considerable freedom in choosing the remaining
values to satisfy the four rules above.  
We refer to the following equations involving 
$D_j^{(k)}, U_j^{(k)}, L_j^{(k)}$, and $\delta_j^{(k)}$
as the \emph{default settings} of these parameters. 
Unless stated otherwise, the proofs of Lemmas~\ref{lem:color-lg} and \ref{lem:color-top-lg} use the default settings.

\begin{align*}
D_j^{(k)} &= \beta \delta_j^{(k-1)}  \cdot D_j^{(k-1)}, &&& 
U_j^{(k)} &= \beta \delta_j^{(k-1)}  \cdot U_j^{(k-1)}, \\
L_j^{(k)} &=  \delta_j^{(k-1)}  \cdot L_j^{(k-1)},  &&&
\delta_j^{(k)} &= \frac{D_j^{(k)} \log{\left(L_j^{(k)} / D_j^{(k)}\right)}}{L_j^{(k)}}.
\end{align*}

\paragraph{Validity of Parameters.}
Before the first iteration the invariants 
$\mathcal{D}_j^{(1)}$,  $\mathcal{U}_j^{(1)}$, and  $\mathcal{L}_j^{(1)}$ are met initially, for each cluster $S_j$.
Suppose $S_j$ is a layer-$i$ large block.
Lemma~\ref{lem:cluster-property} shows that the initial value
of $D_j^{(1)}$ is a valid upper bound on the 
external degree (at most $\epsilon_i\Delta$) 
and anti-degree (at most $3\epsilon_i\Delta$).
We also have 
\[
U_j^{(1)} = (1 + 3 \epsilon_i)\Delta \geq
|S_j| \geq \frac{\Delta}{\log (1/\epsilon_i)} = L_j^{(1)},
\]
where the lower bound is from the definition of \emph{large}
and the upper bound is from Lemma~\ref{lem:cluster-property}.

For $k > 1$, the invariants $\mathcal{D}_j^{(k)}$ and  $\mathcal{U}_j^{(k)}$ might not hold naturally.
Before the $k$th iteration begins
we \emph{forcibly} restore them
by removing from consideration
all vertices in the clusters that violate either invariant, 
putting these vertices in $V_{\bad}$.
Notice that 
\dense\ (version 2) \emph{always} satisfies invariant $\mathcal{L}_j^{(k)}$.

\paragraph{Maintenance of Invariants.}
We calculate the probability for the invariants  
$\mathcal{D}_j^{(k+1)}$ and  $\mathcal{U}_j^{(k+1)}$ 
to naturally hold at a cluster $S_j$.
In what follows, we analyze the $k$th iteration of the algorithm, and assume that $\mathcal{D}_j^{(k)}$ and  $\mathcal{U}_j^{(k)}$ hold initially.
Let $T \subseteq S$ be a set of vertices that are uncolored at the beginning of the $k$th iteration, and
suppose $\delta_j^{(k)} = \max_{j': S_{j'} \cap T \neq \emptyset} \delta_{j'}^{(k)}$.
By Lemma~\ref{lem:dense-101}, after the $k$th iteration, the probability that the number of uncolored vertices in $T$ is at least
$t$ is at most ${|T| \choose t} \cdot \left( O(\delta_j^{(k)})\right)^{t}$.
Using this result, we derive the following bounds:
\begin{align*}
\Prob\mathopen{}\left[\mathcal{U}_j^{(k+1)}\mathclose{}\right] &\geq 1 - \exp\left(-\Omega(U_j^{(k+1)})\right),\\
\Prob\mathopen{}\left[\mathcal{D}_j^{(k+1)}\mathclose{}\right] &\geq 1 - O\mathopen{}\left(U_j^{(k)}\mathclose{}\right) \exp\left(-\Omega(D_j^{(k+1)})\right).
\end{align*}
We first consider $\Prob\mathopen{}\left[\mathcal{U}_j^{(k+1)}\mathclose{}\right]$.
We choose $T$ as the set of uncolored vertices in $S_j\setminus V_{\bad}$ at the beginning of the $k$th iteration, and set $t = U_j^{(k+1)}$.
We have $t = U_j^{(k+1)} = \beta \delta_j^{(k)}  \cdot U_j^{(k)} \geq \beta \delta_j^{(k)} |T|$, and this implies
$\delta_j^{(k)} |T| / t \leq 1/\beta$.
If we select $\beta$ to be a large enough constant, then
\[
1 - \Prob\mathopen{}\left[\mathcal{U}_j^{(k+1)}\mathclose{}\right] 
\leq {|T| \choose t} \cdot \left( O(\delta_j^{(k)})\right)^{t}
\leq \left( O\mathopen{}\left(\delta_j^{(k)}\mathclose{}\right) \cdot e |T| / t\right)^{t}
\leq \left( O(1/\beta) \right)^{t}
= \exp\left(-\Omega\mathopen{}\left(U_j^{(k+1)}\mathclose{}\right)\right).
\]
Next, consider $\Prob\mathopen{}\left[\mathcal{D}_j^{(k+1)}\mathclose{}\right]$.
For each vertex $v \in S_j\setminus V_{\bad}$ that is uncolored at the beginning of the $k$th iteration, define $\mathcal{E}_v^a$ (resp., $\mathcal{E}_v^e$) as the event that the anti-degree (resp., external degree) of $v$ at the end of the $k$th iteration is higher than $D_j^{(k+1)}$.
If we can show that both $\Prob[\mathcal{E}_v^a]$ and $\Prob\mathopen{}\left[\mathcal{E}_v^e\mathclose{}\right]$ are at most $\exp\left(-\Omega\mathopen{}\left(D_j^{(k+1)}\mathclose{}\right)\right)$, then we conclude $\Prob[\mathcal{D}_j^{(k+1)}] \geq 1 - O\mathopen{}\left(U_j^{(k)}\mathclose{}\right) \exp\left(-\Omega(D_j^{(k+1)})\right)$ by a union bound over at most $U_j^{(k)}$ vertices $v \in S_j\setminus V_{\bad}$ that are uncolored at the beginning of the $k$th iteration.

We show that $\Prob[\mathcal{E}_v^e] \leq \exp\left(-\Omega\mathopen{}\left(D_j^{(k+1)}\mathclose{}\right)\right)$.
We choose $T$ as the set of uncolored vertices in  
$\Nstar(v) \setminus (S_j\cup V_{\bad})$ 
at the beginning of the $k$th iteration, and set $t = D_j^{(k+1)}$.
Since the layer number of each vertex in $\Nstar(v) \setminus (S_j\cup V_{\bad})$ is smaller than or equal to the layer number of $S_j$, our requirement about the shrinking rates implies that $\delta_j^{(k)}  \geq \max_{j': S_{j'} \cap T \neq \emptyset} \delta_{j'}^{(k)}$.

We have $t = D_j^{(k+1)} = \beta \delta_j^{(k)}  \cdot D_j^{(k)} \geq \beta \delta_j^{(k)} |T|$, and this implies
$\delta_j^{(k)} |T| / t \leq 1/\beta$.
If we select $\beta$ to be a large enough constant, then
\[
\Prob[\mathcal{E}_v^e] \leq {|T| \choose t} \cdot \left( O\mathopen{}\left(\delta_j^{(k)}\mathclose{}\right)\right)^{t}
\leq \left( O\mathopen{}\left(\delta_j^{(k)}\mathclose{}\right) \cdot e |T| / t\right)^{t}
\leq \left( O(1/\beta) \right)^{t}
 = \exp\left(-\Omega\mathopen{}\left(D_j^{(k+1)}\mathclose{}\right)\right).
\]
The bound $\Prob[\mathcal{E}_v^a] \leq \exp\left(-\Omega\mathopen{}\left(D_j^{(k+1)}\mathclose{}\right)\right)$ is proved in the same way.
Based on the probability calculations above, we are now prepared to
prove Lemmas~\ref{lem:color-lg} and \ref{lem:color-top-lg}.

\medskip 

\begin{proof}[Proof of Lemma~\ref{lem:color-lg}]
We perform $6$ iterations of \dense\ (version 2) using the default settings of all parameters.
Recall that the shrinking rate for the $k$th iteration is 
$\delta_j^{(k)} = \frac{D_j^{(k)} \log\left(L_j^{(k)} / D_j^{(k)}\right)}{L_j^{(k)}}$ for each cluster $S_j$.
If $S_j$ is a layer-$i$ block, we have 
$\delta_j^{(k)} = O\mathopen{}\left(\epsilon_i \log^2 (1/\epsilon_i)\mathclose{}\right)$ 
for each $k\in [1,6]$ since $D_j^{(\cdot)}$ and $L_j^{(\cdot)}$ 
decay at the same rate, asymptotically.

Consider any vertex $v \in V^\star$, and a layer number $i\in[2,\ell]$.
Let $T$ be the set of layer-$i$ neighbors of $v$ in $S$.
To prove Lemma~\ref{lem:color-lg}, it suffices to show that after 6 iterations of \dense~(version 2), with probability $1 - \exp(-\Omega(\poly(\Delta)))$, the number of uncolored vertices in 
$T$ is at most $\epsilon_i^5 \Delta$.

Define $(t_k)$ as in the proof of Lemma~\ref{lem:color-sm}.
\begin{align*}
    t_1 &= |T|,\\
\mbox{ and }\,  t_k &= \max\left\{\beta \delta_j^{(k-1)} t_{k-1},\; \epsilon_i^5 \Delta\right\}.
\end{align*}
Here $\delta_j^{(k)}$ is the common shrinking rate of any layer-$i$
cluster $S_j$.  We have $t_7 = \epsilon_i^5 \Delta$ since 
$\epsilon_i \leq \epsilon_\ell$ is sufficiently small.

Assume that at the beginning of the $k$th iteration, the number of uncolored vertices in $T\setminus V_{\bad}$ is at most $t_k$, and the invariants $\mathcal{D}_j^{(k)}$, $\mathcal{L}_j^{(k)}$, and  $\mathcal{U}_j^{(k)}$ are met for each cluster $S_j$ such that $S_j \cap T \neq \emptyset$.
By Lemma~\ref{lem:dense-101}, after the $k$th iteration, the 
probability that the number of uncolored vertices in $T\setminus V_{\bad}$ is more than $t_{k+1}$ is
\[
{t_{k} \choose t_{k+1}} \cdot \left( O\mathopen{}\left(\delta_j^{(k)}\mathclose{}\right)\right)^{t_{k+1}}
\leq \left( O\mathopen{}\left(\delta_j^{(k)}\mathclose{}\right) \cdot e t_{k} / t_{k+1}\right)^{t_{k+1}}
\leq \left( O(1/\beta) \right)^{t_{k+1}}
 = \exp(-\Omega(t_{k+1})).
\]
Notice that $\exp(-\Omega(t_{k+1})) \leq \exp(-\Omega(\epsilon_i^5 \Delta)) = \exp(-\Omega(\poly(\Delta)))$.
For the maintenance of the invariants, 
$\mathcal{L}_j^{(k+1)}$ holds with probability 1; 
the probability that the invariants $\mathcal{D}_j^{(k+1)}$ and  $\mathcal{U}_j^{(k+1)}$ are met for all clusters $S_j$ such that $S_j \cap T \neq \emptyset$ is at least $1 - O(|T|) \exp(-\Omega(\poly(\Delta))) = 1 - \exp(-\Omega(\poly(\Delta)))$.
By a union bound over all six iterations, with probability $1 - \exp(-\Omega(\poly(\Delta)))$, the number of uncolored layer-$i$ neighbors of $v$ in $S\setminus V_{\bad}$ is at most $t_7 = \epsilon_i^5 \Delta$.
\end{proof}

\medskip

\begin{proof}[Proof of Lemma~\ref{lem:color-top-lg}]
In the setting of Lemma~\ref{lem:color-top-lg}, we deal with only layer-1 large blocks, and so
$D_1^{(k)} = \cdots = D_g^{(k)}$, $U_1^{(k)} = \cdots = U_g^{(k)}$,
$L_1^{(k)} = \cdots = L_g^{(k)}$, $\delta_1^{(k)} = \cdots = \delta_g^{(k)}$, for each iteration $k$.
For this reason we drop the subscripts.
Our algorithm consists of three phases, as follows. Recall that $c$ is a large enough constant related to the failure probability specified in the statement of Lemma~\ref{lem:color-top-lg}.

\paragraph{The Low Degree Case.} 
    The following algorithm and analysis apply to all values of $\Delta$.  The conclusion
    is that we can color most of $V_1^{\LG}$ such that the probability 
    that any vertex joins $V_{\bad}$ is $\Delta^{-\Omega(c)}$ and all remaining
    uncolored vertices (i.e., $R$) induce a graph with maximum degree $O(c^2)$.  Since the guarantee on $V_{\bad}$ is that it induces components with size $\poly(\Delta)\log n$, this analysis is only appropriate when $\Delta$ is, itself, $\poly\log n$.
    We deal with larger $\Delta$ in \emph{The High Degree Case} and prove that
    the uncolored vertices can be partitioned into $R$ and $X$ with the same guarantee on $R$, and the stronger guarantee that $X$ induces $\poly\log n$-size components, regardless of $\Delta$.

\begin{description}
\item[Phase 1.] The first phase consists of 9 iterations of \dense\ (version 2), using 
the default settings of all 
parameters.
Due to the fact that $\epsilon_1 = \Delta^{-1/10}$, we have 
$\delta^{(k)} = O(\Delta^{-1/10} \log^2 \Delta)$ for each $k\in [1,9]$. 
Therefore, at the end of the 9th iteration, we have the parameters
\begin{align*}
D^{(10)} &= \Theta(\log^{18} \Delta),\\
L^{(10)} &= \Theta(\Delta^{1/10} \log^{17} \Delta),\\ 
\mbox{ and } U^{(10)} &= \Theta(\Delta^{1/10} \log^{18} \Delta).
\end{align*}
In view of the previous calculations,
the probability that all invariants hold
for a specific cluster $S_j$ and all $k\in[1,10]$
is at least $1 - \exp(-\Omega(\log^{18} \Delta))$. 
If a cluster $S_j$ does not satisfy an invariant for some $k$, 
then \underline{\emph{all}} 
vertices in $S_j$ halt and join $V_{\bad}$.
They do not participate in the $k$th iteration
or subsequent steps.
\item[Phase 2.] For the 10th iteration, 
we switch to a \emph{non}-default shrinking rate
\begin{align*}
\delta^{(10)} &= \Delta^{-1/20}.
\intertext{However, we still define}
U^{(11)} &=  \beta \delta^{(10)}  \cdot U^{(10)} = \Theta(\Delta^{1/20} \log^{18} \Delta)\\
\mbox{ and }L^{(11)} &=   \delta^{(10)}  \cdot L^{(10)} = \Theta({\Delta^{1/20}} \log^{17} \Delta)
\intertext{
according to their default setting. 
    Since $\beta \delta^{(10)}  \cdot D^{(10)} = o(1)$, we should not adopt the default definition of $D^{(11)}$. Instead, we fix it
    to be the sufficiently large constant $c$.}
D^{(11)} &= c. 
\end{align*}
Using the previous probability calculations, for each cluster $S_j$ the invariant $\mathcal{U}^{(11)}$ holds with probability at least $1 - \exp(-\Omega({\Delta^{1/20}} \poly \log \Delta))$, and the invariant  $\mathcal{L}^{(11)}$ holds with certainty. We will show that for a given cluster $S_j$, the probability that $D^{(11)}$ is a valid degree bound (i.e.,  $\mathcal{D}^{(11)}$ holds) is at least $1 - \Delta^{-\Omega(c)}$. If a cluster $S_j$ does not meet at least one of $\mathcal{U}^{(11)}$, $\mathcal{L}^{(11)}$, or $\mathcal{D}^{(11)}$, then all vertices in $S_j$ halt and join $V_{\bad}$.
\item[Phase 3.] For the 11th iteration, we use the default shrinking rate
\[
\delta^{(11)} = \frac{D^{(11)} \log(L^{(11)} / D^{(11)})}{L^{(11)}} = \Theta\mathopen{}\left(\frac{1}{{\Delta^{1/20}} \log^{16} \Delta}\mathclose{}\right). 
\]
We will show that after the 11th iteration, for each cluster $S_j$, with probability at least $1 - \Delta^{-\Omega(c)}$, there are at most $c^2$ uncolored vertices $v \in S_j$ such that there is at least one uncolored vertex in $\Nout(v) \setminus S_j$. If $S_j$ does not satisfy this property, we put all remaining uncolored vertices in $S_j$ to $V_{\bad}$. For each cluster $S_j$ satisfying this property, in $O(1)$ additional rounds we color all vertices in $S_j$ but $c^2$ of them since at most $c^2$ have potential conflicts outside of $S_j$. 
At this point, the remaining uncolored vertices 
$R$ induce a subgraph of maximum degree at most 
$c^2 + D^{(10)} = c^2 + c = O(c^2)$.
\end{description}

The choice of parameters are summarized as follows. Note that we use the default shrinking rate $\delta^{(i)} = \frac{D^{(i)} \log(L^{(i)} / D^{(i)})}{L^{(i)}}$ for all $i$ except $i = 10$.

\medskip

\begin{center}
\begin{tabular}{ l | l | l | l | l }
   & $D^{(i)}$ & $L^{(i)}$ & $U^{(i)}$ & $\delta^{(i)}$\\\hline\hline
  $i \in [9]$ & $\Theta\mathopen{}\left(\Delta^{\frac{10-i}{10}}\log^{2i-2}\Delta\mathclose{}\right)$ & 
  $\Theta\mathopen{}\left(\Delta^{\frac{11-i}{10}}\log^{2i-3}\Delta\mathclose{}\right)$ & 
  $\Theta\mathopen{}\left(\Delta^{\frac{11-i}{10}}\log^{2i-2}\Delta\mathclose{}\right)$ & 
  $\Theta\mathopen{}\left(\Delta^{-\frac{1}{10}}\log^2 \Delta\mathclose{}\right)$ \istrut[3]{5}\\\hline
  $i = 10$ & $\Theta\mathopen{}\left(\log^{18}\Delta\mathclose{}\right)$ & 
  $\Theta\mathopen{}\left(\Delta^{\frac{1}{10}}\log^{17}\Delta\mathclose{}\right)$ & 
  $\Theta\mathopen{}\left(\Delta^{\frac{1}{10}}\log^{18}\Delta\mathclose{}\right)$ & 
  $\Delta^{-\frac{1}{20}}$ \istrut[3]{5}\\\hline
  $i = 11$  & 
  $c$ & 
  $\Theta\mathopen{}\left(\Delta^{\frac{1}{20}}\log^{17}\Delta\mathclose{}\right)$ & 
  $\Theta\mathopen{}\left(\Delta^{\frac{1}{20}}\log^{18}\Delta\mathclose{}\right)$ & 
  $\Theta\mathopen{}\left(\Delta^{-\frac{1}{20}}\log^{-16} \Delta\mathclose{}\right)$ \istrut[3]{5}\\\hline\hline
\end{tabular}
\end{center}

\paragraph{Analysis of Phase 2.}
Recall $\delta^{(10)} = {\Delta^{-1/20}}$ 
and $D^{(10)} = \Theta(\log^{18} \Delta)$.
By Lemma~\ref{lem:dense-101}, the probability that
the external degree or anti-degree of $v \in S_j$ is at most $c$ is:
\[
1 - {D^{(10)} \choose c} \left(O\mathopen{}\left(\delta^{(10)}\mathclose{}\right)\right)^c \geq
 1 - {{O\mathopen{}\left(\log^{18} \Delta\mathclose{}\right)} \choose c} \left(O\mathopen{}\left(\Delta^{-1/20}\mathclose{}\right)\right)^c \geq
   1 - \Delta^{-\Omega(c)}.
\]
By a union bound over at most $U^{(10)} = \Theta(\Delta^{1/10} \log^{18} \Delta)$ vertices $v \in S_j$ that are uncolored at the beginning of the 10th iteration, the parameter setting $D^{(11)} = c$ is a valid upper bound of  external degree and anti-degree for $S_j$ after the 
10th iteration with probability at least $1 - \Delta^{-\Omega(c)}$.

\paragraph{Analysis of Phase 3.}
Consider a vertex $v \in S_j$ that is uncolored at the beginning of the 11th iteration.
Define the event $\mathcal{E}_v$ as follows. The event $\mathcal{E}_v$ occurs if, 
after the 11th iteration, $v$ is still uncolored, 
and there is at least one uncolored vertex in $\Nout(v) \setminus (S_j\cup V_{\bad})$.
Our goal is to show that the number of vertices $v \in S_j$ such that  $\mathcal{E}_v$ occurs is at most $c^2$ with probability at least
$1 - \Delta^{-\Omega(c)}$.

Consider any size-$c^2$ subset $Y$ of $S_j$.
As a consequence of Lemma~\ref{lem:dense-101}, 
we argue that the probability that 
$\mathcal{E}_v$ occurs for all $v \in Y$ 
is at most
\[
\left(D^{(11)}\right)^{c^2}  \cdot \left( O\mathopen{}\left(\delta^{(11)}\mathclose{}\right)\right)^{c^2\left(1+ 1/D^{(11)}\right)}.
\]
The reason is as follows.
Pick some $v \in Y$.
If $\mathcal{E}_v$ occurs, then there must exist a neighbor 
$v' \in \Nout(v) \setminus (S_j\cup V_{\bad})$ that is uncolored.
The number of uncolored vertices in $\Nout(v) \setminus (S_j\cup V_{\bad})$ 
at the beginning of the 11th iteration is at most $D^{(11)}$,
so there are at most $(D^{(11)})^{c^2}$ ways of mapping each 
$v \in Y$ to a vertex $v' \in \Nout(v) \setminus (S_j\cup V_{\bad})$ of $v$.
Let $T = \bigcup_{v \in Y} \{v, v'\}$.
A vertex outside of $S_j$ can be adjacent to at most 
$D^{(11)}$ vertices in $S_j$, and so 
$|T| \geq c^2(1 + 1/D^{(11)})$.
The probability that all vertices in $T$ are uncolored is 
$\left( O(\delta^{(11)})\right)^{c^2\left(1+ 1/D^{(11)}\right)}$ by Lemma~\ref{lem:dense-101}.
By a union bound over at most $\left(D^{(11)}\right)^{c^2}$ choices of $T$, we obtain the desired probabilistic bound.

Recall that  
$U^{(11)} =  \Theta(\Delta^{1/20} \log^{18} \Delta) 
= L^{(11)} \cdot \Theta(\log \Delta)$ and 
$L^{(11)} = \Theta({\Delta^{1/20}} \log^{17} \Delta)$
are the cluster size upper bound and lower bound at the beginning of the 11th iteration.
By a union bound over at most $\left(U^{(11)}\right)^{c^2}$ choices of a size-$c^2$ subset of $S_j$,
the probability $f$ that there exists $c^2$ vertices $v \in S_j$ such that  $\mathcal{E}_v$ occurs is 
\[
f = \left(U^{(11)}\right)^{c^2} \cdot \left(D^{(11)}\right)^{c^2} \cdot \left( O\mathopen{}\left(\delta^{(11)}\mathclose{}\right)\right)^{c^2\left(1+ 1/D^{(11)}\right)}.
\]
Recall that $D^{(11)} = c$ is sufficiently large. We have 
\begin{align}
\left(U^{(11)}\right)^{c^2} &=  \left( O\mathopen{}\left(L^{(11)} \log \Delta\mathclose{}\right)\right)^{c^2}, 
\label{eqn:f1}\\
\left(D^{(11)}\right)^{c^2} &= O(1), \label{eqn:f2}\\
\left( O\mathopen{}\left(\delta^{(11)}\mathclose{}\right)\right)^{c^2\left(1+ 1/D^{(11)}\right)} &= 
 \left( O\mathopen{}\left(\fr{\log (L^{(11)})}{L^{(11)}}\mathclose{}\right)\right)^{c^2 + c}.\label{eqn:f3}
\end{align}
where $L^{(11)} = \Theta({\Delta^{1/20}} \log^{8} \Delta)$.
Taking the product of (\ref{eqn:f1}), (\ref{eqn:f2}), and (\ref{eqn:f3}), we have:
\[
f = \Theta(\log\Delta)^{O(c^2)} \cdot \Theta\mathopen{}\left(\Delta^{-1/20}\mathclose{}\right)^c = \Delta^{-\Omega(c)},
\]
as required.

\begin{remark}
The analysis of Phase 2 would proceed in the same way if we
had chosen $\delta^{(10)}$ according to its default setting of $\Theta(\Delta^{-1/10}\log^2\Delta)$.  We choose a larger 
value of $\delta^{(10)}$ in order to keep $L^{(11)}$ artificially
large ($\Delta^{\Omega(1)}$), and thereby allow Phase 3 to fail
with smaller probability $\Delta^{-\Omega(c)}$.
\end{remark}

\paragraph{The High Degree Case.} The Low Degree Case handles all $\Delta$ that
are $\poly\log n$.  We now assume $\Delta$ is sufficiently large, i.e., $\Delta > \log^{\alpha c} n$, where $\alpha$ is some large universal constant, and we want to design an algorithm such that no 
vertex joins $V_{\bad}$, and all uncolored vertices are partitioned into $R$ and $X$,
with $R$ having the same $O(c^2)$-degree guarantee as before, and the components
induced by $X$ have size $\log^{O(c)} n = \poly\log n$, regardless of $\Delta$.
Intuitively, the proof follows the same lines as the Low Degree Case, 
but in Phase 1 we first reduce the maximum degree to $\Delta' = \log^{O(c)} n$
then put any bad vertices that fail to satisfy an invariant into $X$ (rather than $V_{\bad}$).  According to the shattering lemma 
(Lemma~\ref{lem:shatter}), the components induced by $X$ have size
$\poly(\Delta')\log n = \log^{O(c)} n$. 
The High Degree Case consists of 13 iterations of \dense{} (version 2)
with the following parameter settings.

\begin{center}
\begin{tabular}{ l | l | l | l | l }
   & $D^{(i)}$ & $L^{(i)}$ & $U^{(i)}$ & $\delta^{(i)}$\\\hline\hline
  $i \in [9]$ & $\Theta\mathopen{}\left(\Delta^{\frac{10-i}{10}}\log^{2i-2}\Delta\mathclose{}\right)$ & 
  $\Theta\mathopen{}\left(\Delta^{\frac{11-i}{10}}\log^{2i-3}\Delta\mathclose{}\right)$ & 
  $\Theta\mathopen{}\left(\Delta^{\frac{11-i}{10}}\log^{2i-2}\Delta\mathclose{}\right)$ & 
  $\Theta\mathopen{}\left(\Delta^{-\frac{1}{10}}\log^2 \Delta\mathclose{}\right)$ \istrut[3]{5}\\\hline
  $i = 10$ & $\Theta\mathopen{}\left( \max\{\log^{18} \Delta, \log n\}\mathclose{}\right)$ & 
  $\Theta\mathopen{}\left(\Delta^{\frac{1}{10}}\log^{17}\Delta\mathclose{}\right)$ & 
  $\Theta\mathopen{}\left(\Delta^{\frac{1}{10}}\log^{18}\Delta\mathclose{}\right)$ & 
  $\Delta^{-\frac{1}{20}}\log^{-18}\Delta$ \istrut[3]{5}\\\hline
  $i = 11$ & $\Theta\mathopen{}\left(\log n\mathclose{}\right)$ & 
  $\Theta\mathopen{}\left(\Delta^{\frac{1}{20}} / \log \Delta\mathclose{}\right)$ & 
  $\Theta\mathopen{}\left(\Delta^{\frac{1}{20}}\mathclose{}\right)$ & 
  $\Delta^{-\frac{1}{20}} \log^{5c} n$ \istrut[3]{5}\\\hline
  $i = 12$ & $\Theta\mathopen{}\left(\log n\mathclose{}\right)$ & 
  $\displaystyle\Theta\mathopen{}\left(\frac{\log^{5c} n}{\log \Delta}\mathclose{}\right)$ & 
  $\Theta\mathopen{}\left(\log^{5c} n\mathclose{}\right)$ & 
  ${\log^{-3c}n}$ \istrut[4]{6}\\\hline
  $i = 13$  & 
  $c$ & 
  $\displaystyle\Theta\mathopen{}\left(\frac{\log^{2c} n}{\log \Delta}\mathclose{}\right)$ & 
  $\Theta\mathopen{}\left(\log^{2c} n\mathclose{}\right)$ & 
  $\displaystyle\Theta\mathopen{}\left(\frac{\log \Delta\log\log n}{\log^{2c} n} \mathclose{}\right)$ \istrut[4]{6}\\\hline\hline
\end{tabular}
\end{center}

We use the default shrinking rate $\delta^{(i)} = \frac{D^{(i)} \log(L^{(i)} / D^{(i)})}{L^{(i)}}$ for all $i$ except $i \in \{10, 11, 12\}$. 
Phase~1 consists of all iterations $i \in [11]$; 
Phase~2 consists of iteration $i = 12$; 
Phase~3 consists of iteration $i = 13$. 
The algorithm and the analysis are similar to the small degree case, so in subsequent discussion we only point out the differences.
In order to have all $\delta^{(i)} \ll 1$, 
we need to have $\Delta^{1/20} \gg \log^{5c} n$. 
We proceed under the assumption that 
$\Delta > \log^{\alpha c} n$ ($\alpha$ is some large universal constant), so this condition is met.

\paragraph{Phase~1.}   In view of previous calculations, 
all invariants hold for a cluster $S_j$ ($\mathcal{U}^{(i)}$, $\mathcal{L}^{(i)}$, and $\mathcal{D}^{(i)}$, 
for $i\in[1,12]$) with probability at least 
$1 - \exp(-\Omega(\log n)) = 1 - 1/\poly(n)$, 
since all parameters $D^{(i)}$, $L^{(i)}$, and $U^{(i)}$  
are chosen to be $\Omega( \log n )$. Therefore, 
no cluster  $S_j$  is put in $V_{\bad}$ 
due to an invariant violation, w.h.p.

\paragraph{Phase~2.} 
Consider iteration $i = 12$. It is straightforward that the invariants $\mathcal{U}^{(13)}$ and $\mathcal{L}^{(13)}$ hold w.h.p., since  $L^{(13)} = \Omega(\log n)$ and $U^{(13)} = \Omega(\log n)$.
Now we consider the invariant $\mathcal{D}^{(13)}$. By Lemma~\ref{lem:dense-101}, the probability that
the external degree or anti-degree of $v \in S_j$ is at most $c$ is:
\[
1 - {D^{(12)} \choose c} \left(O\mathopen{}\left(\delta^{(12)}\mathclose{}\right)\right)^c \geq
 1 - {O(\log n) \choose c} \left(O\mathopen{}\left(\log^{-3c} n\mathclose{}\right)\right)^c \geq
   1 - \left(\log n \right)^{-\Omega(c^2)}.
\]
This failure probability is \emph{not} small enough to guarantee that $\mathcal{D}^{(13)}$ holds everywhere w.h.p. 
In the high degree case, 
if a vertex $v$ belongs to a cluster $S_j$ such that  
$\mathcal{D}^{(13)}$ does not hold, we add
the remaining uncolored vertices
in $S_j$ (at most $U^{(12)}=O(\log^{5c} n)$ of them) to $X$.

\paragraph{Phase~3.} 
Similarly, we will show that after the 13th iteration, for each cluster $S_j$, with probability at least $1 -  \left(\log n \right)^{-\Omega(c^2)}$, there are at most $c^2$ uncolored vertices $v \in S_j$ such that there is at least one uncolored vertex in $\Nout(v) \setminus (S_j\cup X)$. If $S_j$ does not satisfy this property, we put all remaining uncolored vertices in $S_j$ to $X$. For each cluster $S_j$ satisfying this property, in one additional round we can color all vertices in $S_j$ but $c^2$ of them. At this point, the remaining uncolored vertices induce a subgraph $R$ of maximum degree at most $c^2 + D^{(13)} = c^2 + c = O(c^2)$.
Following the analysis in the small degree case, the probability that a vertex $v$ is added to $X$ in the $13$th iteration is 
\begin{align*}
    f &= \left(U^{(13)}\right)^{c^2} \cdot \left(D^{(13)}\right)^{c^2} \cdot \left( O\mathopen{}\left(\delta^{(13)}\mathclose{}\right)\right)^{c^2\left(1+  1/D^{(13)}\right)}\\
    &= O\mathopen{}\left(\log^{2c} n\mathclose{}\right)^{c^2} \cdot O(1) \cdot O\mathopen{}\left(\frac{\log \Delta\log\log n}{\log^{2c} n}\mathclose{}\right)^{c^2+c}\\
    &= O\mathopen{}\left((\log n)^{-2c^2} \cdot (\log \Delta\log\log n)^{c^2+c}\mathclose{}\right) \\
    &= \left(\log n \right)^{-\Omega(c^2)}.
\end{align*}

\paragraph{Size of Components in $X$.} To bound the size of each connected component of $X$, we use the shattering lemma (Lemma~\ref{lem:shatter}).
Define $G'=(V', E')$ as follows. The vertex set $V'$ consists of all vertices in $S$ that remains uncolored at the beginning of iteration $12$. Two vertices $u$ and $v$ are linked by an edge in $E'$ if (i) $u$ and $v$ belong to the same cluster, or (ii) $u$ and $v$ are adjacent in the original graph $G$.
It is clear that the maximum degree $\Delta'$ of $G'$ is  
$U^{(12)}+ D^{(12)} = O(\log^{5c} n)$.  
In view of the above analysis, the probability of $v \in X$ is $1 -  \left(\log n \right)^{-\Omega(c^2)} =  1 - \left(\Delta'\right)^{-\Omega(c)}$, and this is true even if the random bits outside of a constant-radius neighborhood of $v$ in $G'$ are determined adversarially. Applying Lemma~\ref{lem:shatter} to the graph $G'$, the size of each connected component of $X$ is $O(\poly(\Delta')\log n) =  \log^{O(c)}n$, w.h.p., both in $G'$ and in the original graph $G$, since $G'$ is the result of adding some additional edges to the subgraph of $G$ induced by $V'$.
\end{proof}

\medskip

The reader may recall that the proofs of Lemmas~\ref{lem:color-lg} and \ref{lem:color-top-lg} were based on the veracity of 
Lemma~\ref{lem:dense-101}.  
The remainder of this section is devoted to proving Lemma~\ref{lem:dense-101}, which bounds the probability
that a certain number of vertices remain uncolored by
\dense{} (version 2).  By inspection of the \dense{} (version 2) pseudocode, a vertex in $S_j$ can remain uncolored for two different reasons:
\begin{itemize}
\item it never selects a color, because it is not among the $(1 - \delta_{j})|S_j|$ participating vertices in Step 1, or
\item it selects a color in Step 1, but is later decolored in Step 2 
because of a conflict with some
vertex in $S_{j'}$ with $j' < j$.
\end{itemize}

Lemmas~\ref{lem:color-prob}--\ref{lem:uncolor1} analyze different properties of \dense\ (version 2),
which are then applied to prove Lemma~\ref{lem:dense-101}.
Throughout we make use of the property that every 
$\delta_j < 1/K$ for some sufficiently large $K$.

\begin{lemma}\label{lem:color-prob}
Let $T=\{v_1, \ldots, v_k\}$ be any subset of $S_j$ and
$c_1, \ldots, c_k$ be any sequence of colors.
The probability that $v_i$ selects $c_i$ in \dense\ (version 2), for all $i\in[1,k]$,
is $\left( O\mathopen{}\left(\frac{\log (|S_j|/D_j)}{|S_j|}\mathclose{}\right)\right)^{|T|}$.
\end{lemma}
\begin{proof}
Let $p^\star$ be the probability that, for all $i\in [1,k]$, $v_i$ selects $c_i$.
Let $M = (1 - \delta_{j})|S_j|$ be the number of participating vertices in Step 1.
Notice that if $v_i$ is not among the participating vertices, 
then $v_i$ will not select any color, and thus cannot select $c_i$.
Since we are upper bounding $p^\star$, it is harmless to condition on the event that $v_i$ is a participating vertex.
We write $p_i$ to denote the rank of $v_i \in T$ in the random permutation of $S_j$.

Suppose that the ranks $p_1, \ldots, p_k$ were fixed. Recall that each vertex $v_i \in S_j$ is adjacent to all but at most $D_j$ vertices in $S_j$.
Thus, at the time $v_i$ is considered it must have at least
\begin{align*}
\lefteqn{M - p_i + \delta_{j}|S_j| - D_j}\\
&\geq M - p_i + D_j \log(|S_j|/D_j) - D_j   & \mbox{(constraint on $\delta_j$)}\\
&= (M - p_i) + D_j(\log(|S_j|/D_j) - 1)
\end{align*}
available colors to choose from,
at most one of which is $c_i$.
Thus,
\[
p^\star \leq \mathop{\mathrm{E}}_{p_1, \ldots, p_k}
\left[\prod_{i=1}^k \frac{1}{(M - p_i) + D_j(\log(|S_j|/D_j) - 1)}\right].
\]
We divide the analysis into two cases: (i) $k \geq M/2$ and (ii) $k < M/2$. 
For the case $k \geq M/2$, regardless of the choices of $p_1, \ldots, p_k$, we always have
\[
\prod_{i=1}^k \frac{1}{(M - p_i) + D_j(\log(|S_j|/D_j) - 1)} \leq \frac{1}{k!} 
= \left(O(1/k)\right)^{k}
 \leq  \left(O(1/|S_j|)\right)^{|T|}.
\]

We now turn to the case $k < M/2$.
We imagine choosing the rank vector $(p_1,\ldots,p_k)$ one element at a time.
Regardless of the values of $(p_1, \ldots, p_{i-1})$, we always have
\begin{align*}
\lefteqn{\Expect\left[\frac{1}{((M - p_i) + D_j(\log(|S_j|/D_j) - 1)} \;\bigg|\; p_1,\ldots,p_{i-1}\right]}\\
&\leq \frac{1}{M - (i-1)}\sum_{x=0}^{M-i} \frac{1}{x + D_j(\log(|S_j|/D_j) - 1)},
\intertext{since there are $M-(i-1)$ choices for $p_i$ and the worst case
is when $\{p_1,\ldots,p_{i-1}\}=\{1,\ldots,i-1\}$.  Observe that the terms
in the sum are strictly decreasing,
which means the average is maximized when $i=k<M/2$ is maximized.  Continuing,}
&\leq \frac{1}{M/2} \sum_{x=0}^{M/2} \frac{1}{x + D_j(\log(|S_j|/D_j) - 1)}
\intertext{The sum is the difference between two harmonic sums, hence}
&= O\paren{\frac{1}{M}\cdot\big(\log M - \log (D_j(\log(|S_j|/D_j) - 1))\big)} \\
&= O\paren{\frac{\log (|S_j|/D_j)}{|S_j|}}, 
& \mbox{since $M=\Theta(|S_j|)$.}
\end{align*}
Therefore, regardless of $k$, 
$p^\star \leq \left( O\mathopen{}\left(\frac{\log (|S_j|/D_j)}{|S_j|}\mathclose{}\right)\right)^{|T|}$, 
as claimed.
\end{proof}

\begin{lemma}\label{lem:uncolor2}
Let $T$ be any subset of $S_j$.
The probability that all vertices in $T$ are decolored in \dense{} (version 2) is $\left( O\mathopen{}\left(\frac{D_j \log (|S_j|/D_j)}{|S_j|}\mathclose{}\right)\right)^{|T|}$, even allowing the colors selected in $S_1, \ldots, S_{j-1}$ to be determined adversarially.
\end{lemma}
\begin{proof}
There are in total at most $D_j^{|T|}$ different color assignments to $T$ that can result in decoloring all vertices in $T$, since each vertex $v \in T \subseteq S_j$ satisfies $|\Nout(v) \setminus S_j| \leq |\Nstar(v) \setminus S_j| \leq D_j$.
By Lemma~\ref{lem:color-prob} (and a union bound over $D_j^{|T|}$ color assignments to $T$) the probability that all vertices in $T$ are 
decolored is 
$D_j^{|T|} \cdot \left( O\mathopen{}\left(\frac{\log (|S_j|/D_j)}{|S_j|}\mathclose{}\right)\right)^{|T|} = \left( O\mathopen{}\left(\frac{D_j \log (|S_j|/D_j)}{|S_j|}\mathclose{}\right)\right)^{|T|}$.
Recall that for each  $v \in T \subseteq S_j$, we have $\Nout(v) \setminus S_j \subseteq \bigcup_{k=1}^{j-1} S_k$, and so whether $v$ is decolored is independent of the random bits in $S_{j+1}, \ldots, S_{g}$. The above analysis (which is based on Lemma~\ref{lem:color-prob}) holds even allowing the colors selected in $S_1, \ldots, S_{j-1}$ to be determined adversarially.
\end{proof}

\begin{lemma}\label{lem:uncolor1}
Let $T$ be any subset of $S_j$.
The probability that all vertices in $T$ do not select a color in Step 1 of \dense{} (version 2) is  $\left( O(\delta_j)\right)^{|T|}$.
The probability only depends on the random bits within $S_j$.
\end{lemma}
\begin{proof}
The lemma follows from the fact that in \dense{} (version 2) a vertex $v \in S_j$ does not participate in Step 1 with probability $\delta_j$,
and the events for two vertices $u,v\in S_j$ to not participate in Step 1 
are negatively correlated.
\end{proof}

\begin{proof}[Proof of Lemma~\ref{lem:dense-101}]
Recall that we assume the clusters $S = \{S_1, \ldots, S_g\}$ 
are ordered in such a way that for any $u \in S_j$, 
we have $\Nout(u) \subseteq \Nstar(u) \subseteq \bigcup_{k=1}^{j} S_k$.
In the proof we expose the random bits of the clusters in
the order $(S_1, \ldots, S_g)$.

Consider any subset $T \subseteq S$.
Let $U = U_1 \cup U_2$ be a size-$t$ subset $U \subseteq T$.
We calculate the probability that all vertices in $U_1$ do not participate in Step 1, and all vertices in $U_2$ are decolored in Step 2.
Notice that there are at most $2^t$ ways of partitioning $U$ into $U_1 \cup U_2$.

We write $U_1^{(j)} = U_1 \cap S_j$.
Whether a vertex $v \in U_1^{(j)}$ fails to select a color only depends on the random bits in $S_j$.
Thus, by Lemma~\ref{lem:uncolor1},
the probability that all vertices in $U_1$ fail to select a color is at most
$\prod_{j=1}^k \left( O(\delta_j)\right)^{\left|U_1^{(j)}\right|} \leq \left( O(\delta)\right)^{\left|U_1\right|}$.
Recall $\delta = \max_{j: S_j \cap T \neq \emptyset} \delta_j$.

We write $U_2^{(j)} = U_2 \cap S_j$.
Whether a vertex $v \in U_2^{(j)}$ is decolored only depends the random bits in $S_1, \ldots, S_j$.
However, regardless of the  random bits in $S_1, \ldots, S_{j-1}$, the probability that all vertices in $U_2^{(j)}$ are decolored is $ \left( O(\delta_j)\right)^{\left|U_1^{(j)}\right|}$ by Lemma~\ref{lem:uncolor2}.
Recall $\delta \geq \delta_j \geq \frac{D_j \log (|S_j|/D_j)}{|S_j|}$.
Thus, the probability that all vertices in $U_2$ are decolored is at most
$\prod_{j=1}^k \left( O(\delta_j)\right)^{\left|U_2^{(j)}\right|} \leq \left( O(\delta)\right)^{|U_2|}$.

Therefore, by a union bound over at most $|T| \choose t$ choices of $U$
and at most $2^t$ ways of partitioning $U$ into $U_1 \cup U_2$, the probability that the number of uncolored vertices in $T$ is at least $t$ is at most
$2^t \cdot {|T| \choose t} \cdot \left( O(\delta)\right)^{t} =  {|T| \choose t} \cdot \left( O(\delta)\right)^{t}$.
This concludes the analysis of \dense{} (version 2).  
\end{proof}

\section{Conclusion}\label{sect:conclusion}

We have presented a randomized
$(\Delta+1)$-list coloring algorithm that requires 
$O(\Detd(\polylog n))$ rounds
of communication, which is \emph{syntactically} close to the 
$\Omega(\Det(\polylog n))$ lower bound implied
by Chang, Kopelowitz, and Pettie~\cite{ChangKP19}. 
Recall that
$\Det$ and $\Detd$ are the deterministic complexities of $(\Delta+1)$-list coloring and $(\deg+1)$-list coloring.
When $\Delta$ is unbounded (relative to $n$), the best known algorithms for $(\Delta+1)$- and $(\deg+1)$-list coloring are the same:
they use Panconesi and Srinivasan's~\cite{PanconesiS96} $2^{O(\sqrt{\log n})}$-time construction of network decompositions.
Even if optimal $(O(\log n), O(\log n))$-network decompositions could be computed \emph{for free}, we still do not know
how to solve $(\Delta+1)$-list coloring faster than $O(\log^2 n)$ time.  Thus, reducing the $\Detd(\polylog n)$ term in our running
time below $O((\log\log n)^2)$ will require a radically new approach to the problem.

It is an open problem to generalize our algorithm to solve the 
$(\deg+1)$-list coloring problem, and here it may be useful to think about
a problem of intermediate difficulty, at least conceptually.  
Define $(\deg+1)$-coloring to be the coloring problem
when $v$'s palette is $\{1,\ldots,\deg(v)+1\}$ (rather than an arbitrary 
set of $\deg(v)+1$ colors).\footnote{We are aware of one application~\cite{AmirKKNP16} in distributed computing where the palettes are fixed in this way.}  
Whether the problem is $(\deg+1)$-coloring
or $(\deg+1)$-\emph{list} coloring, the difficulty is generalizing 
the notion of ``$\epsilon$-friend edge'' and ``$\epsilon$-sparse vertex''
to graphs with irregular degrees.  
See Figure~\ref{fig:degplusone} for an extreme example illustrating the 
difficulty of $(\deg+1)$-list coloring.
Suppose $N(v)$ is partitioned into
sets $S_1,S_2$ with $|S_1|=|S_2|=|N(v)|/2 = s$. 
The graphs induced by $S_1\cup\{v\}$ and $S_2\cup\{v\}$ are $(s+1)$-cliques
and there are no edges joining $S_1$ and $S_2$.  
The palettes of vertices in $S_1$ 
and $S_2$ are, respectively, $[1,s+1]$ and $[s+1,2s+1]$.

\begin{figure}[h]
\begin{center}
\includegraphics[width=.37\linewidth]{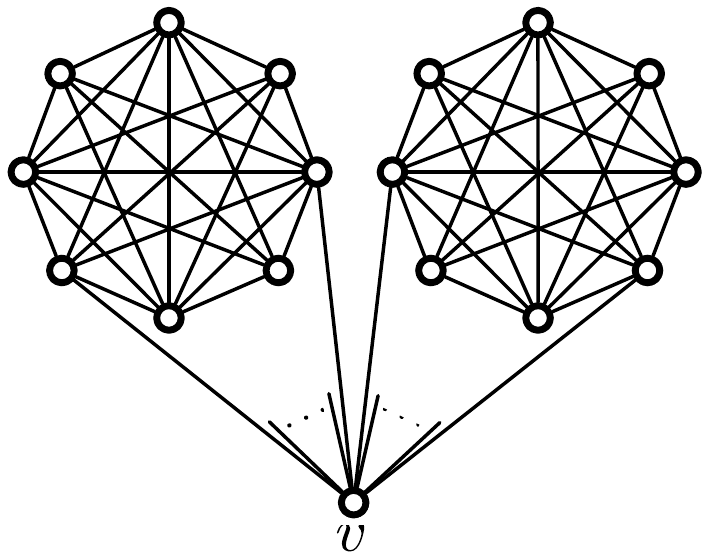}
\end{center}
\caption{An example illustrating the difficulty of $(\deg + 1)$-list coloring.}
\label{fig:degplusone}
\end{figure}


Notice that $v$ is $\epsilon$-sparse according to our definition 
(for any $\epsilon < 1/2$)
and yet regardless of how we design the initial coloring step,
we cannot hope to create more than \emph{one} excess color at $v$
since the two palettes $[1,s+1]\cap [s+1,2s+1]=\{s+1\}$ 
only intersect at one color. 
Thus, it must be wrong to classify $v$
as ``$\epsilon$-sparse'' since it does not satisfy
key properties of $\epsilon$-sparse vertices.
On the other hand, if $v$ is to be classified
as ``$\epsilon$-dense'' then it is not clear whether we can recover any of
the useful properties of $\epsilon$-dense vertices from Lemma~\ref{lem:cluster-property}, e.g., 
that they form almost cliques with $O(1)$ weak diameter and have external degrees bounded by $O(\epsilon\Delta)$.
This particular issue does not arise
in instances of the $(\deg+1)$-coloring problem, 
which suggests that attacking this problem may be a useful 
conceptual stepping stone on the way to solving
$(\deg+1)$-\emph{list} coloring.

\bibliographystyle{abbrv}
\bibliography{references}


\appendix

\section{Concentration Bounds}\label{sect:tools}
We make use of some standard tail bounds~\cite{DubhashiPanconesi09}.
Let $X$ be binomially distributed with parameters $(n,p)$, i.e., it is the sum of $n$ independent 0-1 variables with mean $p$.
We have the following bound on the lower tail of $X$:
\[
\Prob[X \leq t] \leq \exp\left(\frac{-(\mu - t)^2}{2\mu}\right), \; \; \; \text{where } t < \mu = np.
\]

Multiplicative Chernoff bounds give the following tail bounds of $X$ with mean $\mu = np$.
\begin{align*}
\Prob[X \geq (1+\delta)\mu] &\leq
\exp\left(\frac{- \delta^2 \mu}{3}\right)    & & & & \text{if } \delta \in [0,1]\\
\Prob[X \geq (1+\delta)\mu] &\leq \exp\left(\frac{- \delta \mu}{3}\right)    & & & & \text{if } \delta > 1\\
\Prob[X \leq (1-\delta)\mu] &\leq \exp\left(\frac{- \delta^2 \mu}{2}\right)    &  & & & \text{if } \delta \in [0,1]
\end{align*}
Note that Chernoff bounds hold even when $X$ is the summation of $n$ \emph{negatively correlated}
0-1 random variables~\cite{DubhashiR98,DubhashiPanconesi09} with mean $p$,
i.e., total independent is not required.
The bounds for $\Prob[X \geq (1+\delta)\mu]$ also hold when $\mu > np$ is an overestimate of $\Expect[X]$. Similarly, the bound for $\Prob[X \leq (1+\delta)\mu]$ also holds when $\mu <  np$ is an underestimate of $\Expect[X]$.


\medskip

Consider the scenario where $X = \sum_{i=1}^n X_i$, and each $X_i$ is an independent random variable bounded by the interval $[a_i, b_i]$.
Let $\mu = \Expect[X]$.
Hoeffding's inequality~\cite{Hoeffding63}
states that
\[
\Prob[X \geq (1+\delta)\mu] \leq  \exp\left(\frac{-2(\delta \mu)^2}{\sum_{i=1}^n (b_i - a_i)^2}\right).
\]

\section{Proof of Lemma~\ref{lem:initial-color}}\label{sect:oneshot-detail}

In this section, we prove Lemma~\ref{lem:initial-color}.
Fix a constant parameter $p \in (0,1/4)$.
The procedure \oneshot\ is a simple $O(1)$-round coloring procedure that breaks ties by ID.
We orient each edge $\{u,v\}$ towards the endpoint with lower ID, that is, $\Nout(v) = \{u\in N(v) \;|\; \ID(u) < \ID(v)\}$. We assume that each vertex $v$ is associated with a palette $\Psi(v)$ of size $\Delta + 1$, and this is used implicitly in the proofs of the lemmas in this section. 

\begin{framed}
\noindent {\bf Procedure} \oneshot.
\begin{enumerate}
\item Each uncolored vertex $v$ decides to participate independently with probability $p$.
\item Each participating vertex $v$ selects a color $c(v)$ from its palette $\Psi(v)$ uniformly at random.
\item A participating vertex $v$ successfully colors itself if $c(v)$ is not chosen by any vertex in $\Nout(v)$.
\end{enumerate}
\end{framed}

After \oneshot, each vertex $v$ removes all colors from $\Psi(v)$ that are taken by some neighbor $u \in N(v)$.
The number of {\em excess colors} at $v$ is the size of $v$'s remaining palette minus the number of uncolored neighbors of $v$.
We prove one part of Lemma~\ref{lem:initial-color} by showing that after a call to \oneshot, the number of excess colors at
any $\epsilon$-sparse $v$ is $\Omega(\epsilon^2\Delta)$, with probability $1-\exp(-\Omega(\epsilon^2 \Delta))$.
The rest of this section constitutes a proof of Lemma~\ref{lem:initial-color}.

Consider an execution of \oneshot\ with any constant $p\in(0,1/4)$.
Recall that we assume $1/\epsilon \geq K$, for some large enough constant $K$.
Let $v$ be an $\epsilon$-sparse vertex. Define the following two numbers.
\begin{description}
\item $f_1(v)\, :$  the number of vertices $u \in N(v)$ that successfully color themselves by some $c \notin \Psi(v)$.
\item $f_2(v)\, :$  the number of colors $c \in \Psi(v)$ such that at least two vertices in $N(v)$ successfully color themselves  $c$.
\end{description}
It is clear that $f_1(v) + f_2(v)$ is a lower bound on the number of excess colors at $v$ after \oneshot.
Our first goal is to show that $f_1(v) + f_2(v) = \Omega(\epsilon^2 \Delta)$ with probability at least
$1 - \exp(-\Omega(\epsilon^2 \Delta))$.
We divide the analysis into two cases (Lemma~\ref{lem:part1} and Lemma~\ref{lem:part2}), depending on whether
$f_1(v)$ or $f_2(v)$ is likely to be the dominant term.
For any $v$, the preconditions of either Lemma~\ref{lem:part1} or Lemma~\ref{lem:part2} are satisfied.
Our second goal is to show that for each vertex $v$ of degree at least $(5/6)\Delta$, with
high probability, at least $(1 - 1.5p)|N(v)| > (1 - (1.5)/4)\cdot (5/6) \Delta > \Delta/2$ neighbors of $v$ remain uncolored after after \oneshot.
This is done in Lemma~\ref{lem:num-uncolored}.

Lemmas~\ref{lem:aux} and \ref{lem:aux2} establish some generally useful facts about \oneshot,
which are used in the proofs of Lemma~\ref{lem:part1} and \ref{lem:part2}.

\begin{lemma}\label{lem:aux}
Let $Q$ be any set of colors, and let $S$ be any set of vertices with size at most $2\Delta$.
The number of colors in $Q$ that are selected in Step 2
of \oneshot{} by some vertices in $S$ is less than $|Q|/2$ with probability at least $1 - \exp(-\Omega(|Q|))$.
\end{lemma}

\begin{proof}
Let $E_c$ denote the event that color $c$  is selected by at least one vertex in $S$.
Then $\Prob[E_c] \leq \frac{p |S|}{\Delta+1} < 2p < 1/2$, since $p < 1/4$ and $|S| \leq 2\Delta$.
Moreover, the collection of events $\{E_c\}$ are negatively correlated~\cite{DubhashiR98}.

Let $X$ denote the number of colors in $Q$ that are selected by some vertices in $S$.
By linearity of expectation, $ \Expect[X] < 2p \cdot |Q|$.
We apply a Chernoff bound with $\delta = \frac{(1/2) - 2p}{2p}$ and $\mu = 2p \cdot |Q|$.
Recall that $0 < p < 1/4$, and so $\delta > 0$.
For any constant $\delta > 0$, we have:
\[
\Prob[X \geq (1+\delta)\mu = |Q|/2] = \exp(-\Omega(|Q|)). \qedhere
\]
\end{proof}

\begin{lemma}\label{lem:aux2}
Fix a sufficiently small $\epsilon>0$.
Consider a set of vertices $S = \{u_1, \ldots, u_k\}$ with cardinality $\epsilon \Delta/2$.
Let $Q$ be a set of colors such that each $u_i \in S$ satisfies $|\Psi(u_i) \cap Q| \geq (1 - \epsilon/2)(\Delta+1)$.
Moreover, each $u_i \in S$ is associated with a vertex set $R_i$ such that (i) $S \cap R_i = \emptyset$, and (ii) $|R_i| \leq 2\Delta$.
Then, with probability at least $1 - \exp(-\Omega(\epsilon^2 \Delta))$, there are at least $p \epsilon(\Delta+1) / 8$
vertices $u_i \in S$ such that the color $c$ selected by $u_i$ satisfies (i) $c \in Q$, and (ii) $c$ is not selected by any vertex in $R_i \cup S \setminus \{u_i\}$.
\end{lemma}

\begin{proof}
Define $Q_i = \Psi(u_i) \cap Q$. We call a vertex $u_i$ {\em happy} if
$u_i$ selects some color $c\in Q$ and $c$ is not selected by any vertex in $R_i \cup S \setminus \{u_i\}$.
Define the following events.
\begin{description}
\item $E_i^{\text{good}}$:  $u_i$ selects a color $c \in Q_i$ such that $c$ is not selected by any vertices in $R_i$.
\item $E_i^{\text{bad}}$: the number of colors in $Q_i$ that are selected by some vertices in $R_i$ is at least $|Q_i|/2$.
\item $E_i^{\text{repeat}}$: the color selected by $u_i$ is also selected by some vertices in $\{u_1, \ldots, u_{i-1}\}$.
\end{description}
Let $X_i$ be the indicator random variable that \emph{either} $E_i^{\text{good}}$ or $E_i^{\text{bad}}$ occurs, and let $X = \sum_{i=1}^k X_i$.
Let $Y_i$ be the indicator random variable that $E_i^{\text{repeat}}$ occurs, and let $Y = \sum_{i=1}^k Y_i$.
\emph{Assuming that $E_i^{\operatorname{bad}}$ does not occur} for each $i\in [1,k]$,
it follows that $X - 2Y$ is a lower bound on the number of happy vertices.
Notice that by Lemma~\ref{lem:aux}, $\Prob[E_i^{\text{bad}}] = \exp(-\Omega(|Q_i|)) = \exp(-\Omega(\Delta))$.
Thus, assuming that no $E_i^{\text{bad}}$ occurs merely distorts our probability estimates by a negligible $\exp(-\Omega(\Delta))$.
We prove concentration bounds on $X$ and $Y$, which together imply the lemma.

We show that $X \geq p \epsilon \Delta / 7$ with probability $1 - \exp(-\Omega(\epsilon \Delta))$.
It is clear that
\[
\Prob[X_i = 1] \ge \Prob\mathopen{}\left[E_i^{\text{good}} \;|\; \overline{E_i^{\text{bad}}}\mathclose{}\right] \geq \frac{p \cdot |Q_i|/2}{\Delta+1} \geq \frac{p (1 - \epsilon/2)}{2} > \frac{p}{3}.
\]
Moreover, since $\Prob[X_i = 1 \;|\; E_i^{\text{bad}}] = 1$,
the above inequality also holds, when conditioned on \emph{any} colors selected by vertices in $R_i$.
Thus, $\Prob[X \leq t]$ is upper bounded by
$\Prob[\text{Binomial}(n',p') \leq t]$ with $n' = |S| = \epsilon \Delta/2$ and $p' = \frac{p}{3}$.
We set $t = p \epsilon \Delta / 7$. Notice that $n' p' = p \epsilon \Delta / 6 > t$.
Thus, according to a Chernoff bound on the binomial distribution,
$\Prob[X \leq t] \leq \exp(\frac{-(n' p' - t)^2}{2 n'  p'}) = \exp(-\Omega(\epsilon \Delta))$.


We show that $Y \leq p \epsilon^2 \Delta/2$ with probability $1 - \exp(-\Omega(\epsilon^2 \Delta))$. It is clear that $\Prob[Y_i = 1] \leq \frac{p (i-1)}{\Delta+1} \leq \frac{p \epsilon}{2}$,
even if we condition on arbitrary colors selected by vertices in $\{u_1, \ldots, u_{i-1}\}$.
We have $\mu = \Expect[Y] \leq \frac{p \epsilon}{2} \cdot |S| =  \frac{p \epsilon^2 \Delta}{4}$.
Thus, by a Chernoff bound (with $\delta = 1$),
$\Prob[Y \geq p \epsilon^2 \Delta/2] \leq
\Prob[Y \geq (1+\delta)\mu] \leq \exp(-\delta^2 \mu / 3) = \exp(-\Omega(\epsilon^2 \Delta))$.

To summarize, with probability at least $1 - \exp(-\Omega(\epsilon^2 \Delta))$,
we have $X - 2Y \geq p \epsilon \Delta / 7 - 2p\epsilon^2 \Delta/2 > p \epsilon(\Delta+1) / 8$.
\end{proof}

Lemma~\ref{lem:part1} considers the case when a large fraction of $v$'s neighbors are likely to color themselves
with colors outside the palette of $v$, and therefore be counted by $f_1(v)$.
This lemma holds regardless of whether $v$ is $\epsilon$-sparse or not.

\begin{lemma}\label{lem:part1}
Suppose that there is a subset $S \subseteq N(v)$ such that
$|S| = \epsilon \Delta/5$, and
for each $u \in S$, $|\Psi(u) \setminus \Psi(v)| \geq \epsilon (\Delta+1)/5$.
Then $f_1(v) \geq \frac{p \epsilon^2 \Delta}{100}$ with probability at least $1 - \exp(-\Omega(\epsilon^2 \Delta))$.
\end{lemma}
\begin{proof}
Let $S = (u_1, \ldots, u_k)$ be sorted in increasing order by ID.
Define $R_i = \Nout(u_i)$, and $Q_i = \Psi(u_i) \setminus \Psi(v)$.
Notice that $|Q_i| \geq \epsilon \Delta / 5$. Define the following events.
\begin{description}
\item $E_i^{\text{good}}$: $u_i$ selects a color $c \in Q_i$ and $c$ is not selected by any vertex in $R_i$.
\item $E_i^{\text{bad}}$: the number of colors in $Q_i$ that are selected by vertices in $R_i$ is more than $|Q_i|/2$.
\end{description}

Let $X_i$ be the indicator random variable that either $E_i^{\text{good}}$ or $E_i^{\text{bad}}$ occurs, and let $X = \sum_{i=1}^k X_i$.
Given that the events $E_i^{\text{bad}}$ for all $i\in[1,k]$ do not occur, we have $X \leq f_1(v)$,\footnote{In general, $X$ does not necessarily equal $f_1(v)$, since in the calculation of $X$ we only consider the vertices in $S$, which is a subset of $N(v)$.} since
if $E_i^{\text{good}}$ occurs, then $u_i$ successfully colors itself by some color $c \notin \Psi(v)$.
By  Lemma~\ref{lem:aux}, $\Prob[E_i^{\text{bad}}] = \exp(-\Omega(|Q_i|)) = \exp(-\Omega(\epsilon \Delta))$.
Thus, up to this negligible error, we can assume that $E_i^{\text{bad}}$ does not occur, for each $i\in[1,k]$.

We show that $X \geq \epsilon^2 \Delta / 100$ with probability $1 - \exp(-\Omega(\epsilon^2 \Delta))$.
It is clear that $\Prob[X_i = 1] \geq \Prob[E_i^{\text{good}} \;|\; \overline{E_i^{\text{bad}}}] \geq \frac{p |Q_i|/2}{\Delta+1} \geq \frac{p \epsilon}{10}$, and this inequality holds even when conditioning on {\em any} colors selected by vertices in $R_i$ and $\bigcup_{1 \leq j < i} R_j \cup \{u_j\}$. 
Since $S = (u_1, \ldots, u_k)$ is sorted in increasing order by ID, $u_i \notin R_j = \Nout(u_j)$ for any $j \in [1,i)$.
Thus, $\Prob[X \leq t]$ is upper bounded by $\Prob[\text{Binomial}(n',p') \leq t]$ with $n' = |S| = \epsilon \Delta/5$ and $p' = \frac{p \epsilon}{10}$.
We set $t = \frac{n' p'}{2} = \frac{p \epsilon^2 \Delta}{100}$.
Thus, according to a lower tail of the binomial distribution,
$\Prob[X \leq t] \leq \exp\left(\frac{-(n' p' - t)^2}{2 n'  p'}\right) = \exp(-\Omega(\epsilon^2 \Delta))$.
\end{proof}

Lemma~\ref{lem:part2} considers the case that many pairs of neighbors of $v$ are likely to color
themselves the same color, and contribute to $f_2(v)$.   Notice that any $\epsilon$-sparse vertex
that does not satisfy the preconditions of Lemma~\ref{lem:part1} \emph{does} satisfy the
preconditions of Lemma~\ref{lem:part2}.

\begin{lemma}\label{lem:part2}
Let $v$ be an $\epsilon$-sparse vertex.
Suppose that there is a subset $S \subseteq N(v)$ such that
$|S| \geq  (1-\epsilon/5) \Delta$, and
for each $u \in S$, $|\Psi(u) \cap \Psi(v)| \geq (1-\epsilon/5)(\Delta+1)$.
Then $f_2(v) \geq p^3 \epsilon^2 \Delta / 2000$ with probability at least $1 - \exp(-\Omega(\epsilon^2 \Delta))$.
\end{lemma}

\begin{proof}
Let $S'=\{u_1, \ldots, u_k\}$ be any subset of $S$ such that
(i) $|S'| = \frac{p \epsilon \Delta}{100}$,
(ii) for each $u_i \in S'$, there exists a set $S_i \subseteq S \setminus (S' \cup N(u_i))$ of size $\frac{\epsilon \Delta}{2}$.
The existence of $S',S_1,\ldots,S_{k}$ is guaranteed by the $\epsilon$-sparseness of $v$.  In particular, $S$ must contain
at least $\epsilon\Delta - \epsilon\Delta/5 > p\epsilon\Delta/100 = |S'|$ \emph{non}-$\epsilon$-friends of $v$, and for each such non-friend $u_i\in S'$, we have
$|S \setminus (S'\cup N(u_i))| \ge
|S| - |S'| - |N(u_i)| \geq \Delta((1-\epsilon/5) - p\epsilon/100 - (1-\epsilon)) > \epsilon\Delta/2$.

Order the set $S' = \{u_1, \ldots, u_k\}$ in increasing order by vertex ID.
Define 
$Q_i = \Psi(u_i) \cap \Psi(v)$.
Define $Q_i^{\text{good}}$ as the subset of colors $c\in Q_i$ such that
$c$ is selected by some vertex $w \in S_i$, but $c$ is not selected by any vertex in $(\Nout(w) \cup \Nout(u_i)) \setminus S'$.
Define the following events.
\begin{description}
\item $E_i^{\text{good}}$:  $u_i$ selects a color $c \in Q_i^{\text{good}}$.
\item $E_i^{\text{bad}}$:  the number of colors in $Q_i^{\text{good}}$ is less than $p \epsilon(\Delta+1) / 8$.
\item $E_i^{\text{repeat}}$: the color selected by $u_i$ is also selected by some vertices in $\{u_1, \ldots, u_{i-1}\}$.
\end{description}

Let $X_i$ be the indicator random variable that either $E_i^{\text{good}}$ or $E_i^{\text{bad}}$ occurs, and let $X = \sum_{i=1}^k X_i$.
Let $Y_i$ be the indicator random variable that $E_i^{\text{repeat}}$ occurs, and let $Y = \sum_{i=1}^k Y_i$.
Suppose that $E_i^{\text{good}}$ occurs.
Then there must exist a vertex $w \in S_i$ such that both $u_i$ and $w$ successfully color themselves $c$. Notice that $w$ and $u_i$ are not adjacent.
Thus, $X - Y \leq f_2(v)$, given that $E_i^{\text{bad}}$ does not occur, for each $i\in [1,k]$.
Notice that $\Prob[E_i^{\text{bad}}] =  \exp(-\Omega(\epsilon^2 \Delta))$ (by Lemma~\ref{lem:aux2} and the definition of $Q_i^{\text{good}}$), and up to this negligible error we can assume that $E_i^{\text{bad}}$ does not occur. 
In what follows, we prove concentration bounds on $X$ and $Y$, which together imply the lemma.

We show that $X \geq \frac{p^3 \epsilon^2 \Delta}{1000}$ with probability $1 - \exp(-\Omega(\epsilon^2 \Delta))$.
It is clear that $\Prob[X_i = 1] \geq p \cdot \frac{p \epsilon (\Delta+1) / 8}{\Delta+1} = \frac{p^2 \epsilon}{8}$.\footnote{In the calculation of $X$, we first reveal all colors selected by vertices in $V \setminus S'$, and then we reveal the colors selected by $u_1, \ldots, u_k$ in this order. The value of $X_i$ is determined when the color selected by $u_i$ is revealed. Regardless of the colors selected by vertices in $V \setminus S'$ and $\{u_1, \ldots, u_{i-1}\}$, we have $\Prob[X_i = 1] \geq \frac{p^2 \epsilon}{8}$. }  
Thus, $\Prob[X \leq t]$ is upper bounded by $\Prob[\text{Binomial}(n',p') \leq t]$ with $n' = |S'| = \frac{p\epsilon \Delta}{100}$ and $p' = \frac{p^2 \epsilon}{8}$.
We set $t = \frac{p^3 \epsilon^2 \Delta}{1000} < n' p'$.
According to a tail bound of binomial distribution,
$\Prob[X \leq t] \leq \exp(\frac{-(n' p' - t)^2}{2 n'  p'}) = \exp(-\Omega(\epsilon^2 \Delta))$.

We show that $Y \leq \frac{p^3 \epsilon^2 \Delta}{2000}$ with probability $1 - \exp(-\Omega(\epsilon^2 \Delta))$.
It is clear that $\Prob[Y_i = 1] \leq p \cdot \frac{(i-1)}{\Delta+1} \leq \frac{p^2 \epsilon}{100}$ regardless of the colors selected by vertices in $\{u_1, \ldots, u_{i-1}\}$.
We have $\mu = \Expect[Y] \leq \frac{p^2 \epsilon}{100} \cdot |S'| =  \frac{p^3 \epsilon^2 \Delta}{10,000}$.
Thus, by a Chernoff bound (with $\delta = 4$),
$\Prob[Y \geq \frac{p^3 \epsilon^2 \Delta}{2000}] \leq \Prob[Y \geq (1+\delta)\mu] \leq \exp(-\delta \mu / 3) = \exp(-\Omega(\epsilon^2 \Delta))$.

To summarize, with probability at least $1 - \exp(-\Omega(\epsilon^2 \Delta))$, we have 
$X - Y \geq p^3 \epsilon^2 \Delta / 1000 
 - p^3\epsilon^2 \Delta / 2000 
= p^3 \epsilon^2 \Delta / 2000$.
\end{proof}

\begin{lemma}\label{lem:num-uncolored}
The number of  vertices in $N(v)$ that remain uncolored after \oneshot\ is at least $(1 - 1.5 p)|N(v)|$, with probability at least $1 - \exp(-\Omega(|N(v)|))$.
\end{lemma}
\begin{proof}
Let $X$ be the number of vertices in $N(v)$ participating in \oneshot. It suffices to show that $X \leq  1.5 p |N(v)|$ with probability $1 - \exp(-\Omega(|N(v)|))$. Since a vertex participates with probability $p$,
\begin{align*}
\Prob[X \geq (1 + 1/2)p |N(v)|] &\leq \exp\left(-\frac{(1/2)^2 p |N(v)|}{3}\right) = \exp(-\Omega(|N(v)|))
\end{align*}
by Chernoff bound with $\delta = 1/2$.
\end{proof}

\end{document}